\documentclass[11pt]{article}

\usepackage{amsfonts}
\usepackage{amsmath}
\usepackage[dvips]{graphicx}
\usepackage{graphicx,indentfirst,amsmath}
\usepackage{comment}
\usepackage{booktabs}
\usepackage{color}
\usepackage{soul}
\usepackage{caption}
\usepackage{subcaption}

\usepackage[square,colon]{natbib}

\makeatletter
\def\@biblabel#1{}
\makeatother

\newtheorem{theorem}{Theorem}[section]

\newtheorem{lemma}{Lemma}[section]
\newtheorem{proposition}{Proposition}[section]

\newenvironment{proof}[1][Proof]{\textbf{#1.} }{\ \rule{0.5em}{0.5em}}
\makeatletter

\newcommand{\Rmnum}[1]{\expandafter\@slowromancap\romannumeral #1@}
\makeatother

\def\tr{\;\mathrm{tr}\;}

\def\sign{\mathrm{sign}\, }
\def\<{\langle}
\def\>{\rangle}
\def\G{{\cal G}}

\begin{document}

\title{\bf Bayesian precision matrix estimation for graphical Gaussian models with edge and vertex symmetries}
\author{H\'{e}l\`{e}ne Massam, Qiong Li and Xin Gao  \\
{Department of Mathematics and Statistics, York University, Toronto, Canada}}

\date{}
\maketitle

\noindent {\bf ABSTRACT}: Graphical Gaussian models with edge and vertex symmetries were introduced by \citet{HojLaur:2008} who also gave an algorithm to  compute the maximum likelihood estimate of the precision  matrix for such models. In this paper, we take a Bayesian approach to the estimation of the precision matrix. We consider only those models where the symmetry constraints are imposed on the precision matrix and which thus form a natural exponential family with the precision matrix as the canonical parameter.

 We first identify the Diaconis-Ylvisaker conjugate prior for these models and develop a scheme to sample from the prior and posterior distributions. We thus obtain estimates of the posterior mean of the precision matrix.

 Second, in order to verify the precision of our estimate, we derive the explicit analytic expression of the expected value of the precision matrix when the graph underlying our model is a tree, a complete graph on three vertices and a decomposable graph on four vertices with various symmetries. In those cases, we compare our estimates with the exact value of the mean of the prior distribution. We also verify the accuracy of our estimates of the posterior mean on simulated data for graphs with up to thirty vertices and various symmetries.


\section{Introduction}
Given an undirected  graph $G=(V,E)$ where $V$ is the set of vertices and $E\subset V\times V$ is the set of undirected edges denoted $(i,j)$, a graphical Gaussian model is a family of Gaussian distribution for $X=(X_v,\;v\in V)$ where the conditional independences between the components of $X$ can be represented by means of a graph as follows:
$$(i,j)\not \in E\Rightarrow X_i\perp X_j\mid X_{V\setminus\{i,j\}}.$$

The graphical Gaussian models with edge and vertex symmetries, which we here call the colored graphical Gaussian model, have been introduced in  \cite{HojLaur:2008}. These models are defined as graphical Gaussian models with three different types of symmetry constraints: equality of specified entries of the  inverse of the covariance matrix $K$, equality of specified entries of the correlation matrix  or equality of specified entries of $K$ generated by a subgroup of the automorphism group of $G$. These models are called respectively RCON, RCOR and RCOP models. In this paper, we consider only RCON   models which form a natural exponential family with the precision matrix $K$ as the canonical parameter. The model can be represented by colored graphs, where edges or vertices have the same coloring if the corresponding elements of the precision matrix are equal.

 \cite{HojLaur:2008}  proposed an algorithm to compute the maximum likelihood estimates of $K$. However, to the best of our knowledge, there is no work for Bayesian estimates of $K$. Since the RCON model is a natural exponential family, we use the \cite{dy79} (henceforth abbreviated DY) conjugate prior for $K$. This yields a distribution similar to the DY conjugate prior for graphical Gaussian models but with the symmetry constraints mentioned above. We will therefore call this conjugate prior the colored $G$-Wishart.

   Our sampling scheme is an adaptation of the independent
   Metropolis-Hastings (henceforth abbreviated MH) algorithm for the $G$-Wishart proposed by \cite{Mitsa:2011}.
 In the case of regular graphical Gaussian models (not colored), the DY  conjugate prior for $K$ is the so called $G$-Wishart and there are  a number of  sampling schemes for this distribution: see \cite{Pic:2000}, \cite{Mitsa:2011}, \cite{DobLR:2011}, \cite{Wang:2012}, \cite{Len:2013}, the more recent ones being generally more efficient than the preceding ones. However, by the very nature of these sampling schemes, only \cite{Mitsa:2011} and \cite{DobLR:2011} could be adapted to the colored $G$-Wishart. Moreover, we found that adapting \cite{DobLR:2011} leads to significant autocorrelation. The sampling scheme that we propose in Section 3 of this paper is therefore an adaptation to the colored $G$-Wishart of the sampling scheme for the $G$-Wishart given by   \cite{Mitsa:2011}.

In order to judge the accuracy of our sample, we need to either know the exact value of the expected value of the precision matrix or we need to proceed by simulations. We will do both. The RCON models for $X=(X_v, v\in V)$ that we consider in this paper are natural exponential families with density of the form
$$f(X;K)\propto \exp\{ \langle K, XX^t \rangle -\frac{1}{2}\log |K|\}$$
 where  $|K|$ is the determinant of the precision matrix  and $\langle A, B\rangle=\tr AB$ denotes the inner product of the two symmetric matrices $A$ and $B$. Let ${\cal G}$ be the colored version of $G$ to be defined in Section 2. The DY conjugate prior for $K$ is then of the form
 $$p(K;\delta, D)=\frac{1}{I_{\cal G}(\delta,D)}\exp \{-\frac{1}{2}(\langle K, D\rangle-(\delta-2)\log |K|)\}$$
 which is itself an exponential family in $K$ and thus the expected value of $K$ is given by the derivative of $\log I_{\cal G}(\delta,D)$ with respect to the canonical parameter.
 The difficulty is, of course, to compute  the normalizing constant $I_{\cal G}(\delta,D)$.

  Even for the uncolored $G$-Wishart, until recently, one did not know how to obtain the analytic expression of $I_G(\delta,D)$ unless $G$ was decomposable. Using an iterative method and special functions, \citet{uhler:2014} seem to have solved this very difficult problem but their results  do not extend to our colored $G$-Wishart with symmetry constraints which, as we shall see, add another level of difficulty. We  do give, in Section 4, the explicit analytic expression of the normalizing constant of the coloured $G$-Wishart for some special graphs $\G$: general trees,   star graphs, a complete graph on 3 vertices and a simple decomposable model on 4 vertices with various symmetry constraints. For these particular coloured graphs, the analytic expression of   $I_{\G}(\delta,D)$ allows us to verify the accuracy of the estimate of the prior mean obtained with the sampling scheme developed in Section 3 by comparing  it to the true prior mean obtained by differentiation of $I_{\G}(\delta,D)$.
The numerical results are given in Section 5.
 In Section 6, we compute the posterior mean for  two relatively high-dimensional examples where the underlying graphs are  cycles of length 20 and 30. Since we cannot compute the analytic expression of the normalizing constant in those cases, we compare the posterior mean estimate obtained through our sampling method of Section 3 with the true value of $K$ used for simulating the initial Gaussian data. Our  numerical results show good accuracy and small relative errors that decrease with sample size, as expected.

 \section{Preliminaries}
We will now recall  some definitions and concepts that we will need in the sequel.
Let $G=(V,E)$ be an undirected graph as defined in the introduction.
Let $P_G$ be the cone of $p\times p$ positive definite matrices $X$ with entry $X_{ij}=0$  whenever $(i,j)\not \in E$. It is well-known (see \cite{Laur:1996}) that the graphical Gaussian model Markov with respect to $G$ is the set of Gaussian $N(0,\Sigma)$ distributions
\begin{equation}
\label{gm}
{\cal N}_G=\{N(0,\Sigma)\mid K=\Sigma^{-1}\in P_G\}.
\end{equation}
The Diaconis-Ylvisaker conjugate prior for the parameter $K$ is the so-called $G$-Wishart distribution (see \cite{Rov:2002}) defined on $P_G$ and with density
$$p(K|\delta, D)=\frac{1}{I_{G}(\delta, D)}|K|^{(\delta-2)/2}\exp\{-\frac{1}{2} \langle K, D\rangle \},$$
where $\delta > 0$  and $D$, a symmetric positive definite $p\times p$ matrix, are the hyper parameters of the prior distribution on $K$ and $I_{G}(\delta, D)$ is the normalizing constant, namely,
$$I_{G}(\delta, D)=\int_{P_G} |K|^{(\delta-2)/2}\exp\{-\frac{1}{2} \langle K, D\rangle \}dK.$$
Let us now define the RCON model. Let  $\mathcal{V}=\{V_1,\dots,V_k\}$ form a partition  of $V=\{1,\ldots,p\}$ and let $\mathcal{E}=\{E_1,\ldots,E_l\}$ form a partition of the edge set $E$. If all the vertices belonging to an element $V_i$ of  ${\cal V}$ have the same colour, we say that ${\cal V}=\{V_1,\ldots,V_k\}$ is a colouring of $V$. Similarly if all the edges belonging to an element $E_i$ of ${\cal E}$ have the same colour, we say that ${\cal E}$ is a colouring of the edges of $ G$ and that $({\cal V}, {\cal E})$ is a coloured graph.

Consider model (\ref{gm}). If, for $K\in P_{ G}$, we impose the further restrictions that if
\newline $(C_1):\;$  $m$ is a vertex class  in ${\cal V}$, then for all $i\in m$, $K_{ii}$ are equal,
\newline $(C_2):\;$  $s$ is an edge class in $ {\cal E}$, then for all $(i,j)\in s$, the entries $K_{ij}$ of the precision matrix are equal,
 then model (\ref{gm}) becomes a coloured graphical Gaussian model called the RCON$({\cal V, \cal E})$ model.

For the computation of the analytic expression of  $I_{G}(\delta, D)$, we will need two special functions, the Bessel function of the third kind and the hypergeometric function  $_{p}F_{q}$. The Bessel function of the third kind  is defined as
\begin{equation*}
\label{KB}
K_{\lambda}(z)=\int_0^{\infty}u^{2\lambda-1}e^{-\frac{z}{2}(\frac{1}{u^2}+u^2)}du.\end{equation*}
For some special values of $\lambda$,  the Bessel function can be given explicitly
\begin{eqnarray*}
&&K_{1/2}(z)=\sqrt{\frac{\pi}{2}}z^{-1/2}e^{-z},\;\;
K_{3/2}(z)=\sqrt{\frac{\pi}{2}}(z^{-1/2}+z^{-3/2})e^{-z},\;\;\\
&&K_{5/2}(z)=\sqrt{\frac{\pi}{2}}(z^{-1/2}+3z^{-3/2}+3z^{-5/2})e^{-z}.
\end{eqnarray*}
We will also use the classical formula
\begin{equation*}
\label{K}
(\frac{p}{q})^{\frac{\lambda}{2}}K_{\lambda}(\sqrt{pq})=\int_0^{\infty}u^{2\lambda-1}e^{-\frac{1}{2}(\frac{p}{u^2}+qu^2)}du.
\end{equation*}
The hypergeometric function $_{p}F_{q}$ is defined by the power series:
$$_{p}F_{q}(a_{1},\ldots, a_{p};b_{1},\ldots,b_{q};z)=\sum\limits^{\infty}_{k=0}\frac{(a_{1})_{k}\cdots (a_{p})_{k}}{(b_{1})_{k}\cdots (b_{p})_{k}}\frac{z^{k}}{k!}$$
where
\[ (a)_{k} = \left\{
   \begin{array}{l l}
     1 & \quad \text{if $n=0$ }\\
     a(a+1)\cdots (a+n-1) & \quad \text{if $n>0$}\;.
   \end{array} \right.\]
The derivative of the hypergeometric function $_{p}F_{q}(a_{1},\ldots, a_{p};b_{1},\ldots,b_{q};z)$ is given by

\begin{equation}
\label{dhy}
\frac{d}{dz}[_{p}F_{q}(a_{1},\ldots, a_{p};b_{1},\ldots,b_{q};z)]=\frac{a_{1}\cdots a_{p}}{b_{1}\cdots b_{q}}  ( _{p}F_{q}(a_{1}+1,\ldots, a_{p}+1;b_{1}+1,\ldots,b_{q}+1;z))\;.
\end{equation}

\section{The coloured $G$-Wishart distribution: a sampling method }
\subsection{The coloured $G$-Wishart}
For $G$ an undirected graph, let ${\cal G}=({\cal V},{\cal E})$ denote its coloured version as defined in Section 2, and let $P_{\cal G}$ denote the cone of $p\times p$ positive definite matrices in $P_G$ which also obey the symmetry constraints of ${\cal G}$, i.e.
$$P_{\G}=\{K\in P_G\mid (C_1)\;\mbox{and}\; (C_2)\;\mbox{are satisfied}\}\;.$$

 We define the  $CG$-Wishart, i.e. the colored $G$-Wishart, to be the DY-conjugate prior for the parameter $K$ of the RCON$({\cal V, \cal E})$ model. Its density is
\begin{equation}
\label{cg}
p(K|\delta, D)=\frac{1}{I_{\cal G}(\delta, D)}|K|^{(\delta-2)/2}\exp\{-\frac{1}{2} \tr(KD)\}{\bf 1}_{P_{\cal G}}(K)
\end{equation}
where $\delta > 0$  and $D$, a symmetric
$p\times p$ matrix, are  hyper parameters  and $I_{\cal G}(\delta, D)$ is the normalizing constant, namely,
\begin{equation}
\label{norm}
I_{\cal G}(\delta, D)=\int_{P_{\cal G}} |K|^{(\delta-2)/2}\exp\{-\frac{1}{2} tr(KD)\}dK.
\end{equation}

 We will see that $I_{\G}(\delta, D)$ is finite only for $D$ in the dual cone of $P_{\cal G}$ which has to be  determined for each ${\cal G}$. We will derive the dual cones in the special cases that we consider in Section 4.

In this section, following what has been done in \cite{Mitsa:2011}, we want to derive a MH algorithm to sample from the  $CG$-Wishart. But in order to do so, following \cite{Ata:2005}, we want to express the density of the $CG$-Wishart in terms  of the Cholesky components of $K$ scaled by $D$. To this end, we consider the Cholesky decomposition of $D^{-1}$ and $K$ written as
$$D^{-1}=Q^TQ,\;\;\;K=\Phi^{T}\Phi$$
 with $Q=(Q_{ij})_{1\leq i\leq j\leq p}$ and $\Phi=(\Phi_{ij})_{1\leq i\leq j\leq p}$ upper triangular matrixes with real  positive diagonal entries
and we will use the variable
$$\Psi=\Phi Q^{-1}.$$
  We adapt the MH algorithm given in \citet{Mitsa:2011} for the $G$-Wishart to the $CG$-Wishart, using the variable $\Psi$ rather than $K$.  To do so, we first  define
 $$v_{u}(G)=\min\{(i,j):i\leq j, \;(i,j)\in u \in \mathcal{V}\cup \mathcal{E} \}$$
 where the minimum is defined according to the lexicographical order and
 $$v(G)=\bigcup \limits_{u \in \mathcal{V}\cup \mathcal{E}}v_{u}(G).$$
 We will write $K^{v(G)}=(K_{ij}|\;(i,j)\in v(G))$ for the free elements of $K$. The zero  and colouring constraints on the elements of $K$ determine the free entries $\Phi^{v(G)}=\{\Phi_{ij}:(i,j) \in v(G) \}$ and $\Psi^{v(G)}=\{\Psi_{ij}:(i,j) \in v(G) \}$ of the matrices $\Phi$ and $\Psi$ respectively. Each non-free element $\Phi_{ij}$
and $\Psi_{ij}$ with $(i,j) \notin v(G)$ is a function of the free elements $\Phi^{v(G)}$ and $\Psi^{v(G)}$ that precede it in the lexicographical order. The following two propositions give the expression of the non-free entries in function of the free ones and the free entries of $K$. The first part of each proposition can be found in \cite{Rov:2002}.

\begin{proposition}
\label{chokphi}
Let  $K=\Phi^{T}\Phi$ be an element of $P_{\cal G}$. Then the entries $\Phi_{ij}$ are such that
\begin{align}
\Phi_{ij} &=\frac{K_{ij}-\sum\limits^{i-1}_{k=1}\Phi_{ki}\Phi_{kj}}{\Phi_{ii}}, &\mbox{for}\;(i,j)\in v(G)\nonumber\\
\Phi_{1k}&=0, &\mbox{for}\;K_{1k}=0,\;k=2,\ldots,p\label{1}\\
\Phi_{ij}&=-\frac{\sum\limits^{i-1}_{k=1}\Phi_{ki}\Phi_{kj}}{\Phi_{ii}}&\mbox{for}\;K_{ij}=0,\ldots, p,\;  i \neq 1\nonumber\\
\Phi_{ij} &=\frac{\Phi_{i_{u}j_{u}}\Phi_{i_{u}i_{u}}+\sum\limits^{i_{u}-1}_{k=1}\Phi_{ki_{u}}\Phi_{kj_{u}}-\sum\limits^{i-1}_{k=1}\Phi_{ki}\Phi_{kj}}{\Phi_{ii}}&\mbox{for}\;K_{ij}\neq 0,\;(i,j) \in u \in \mathcal{V}\cup \mathcal{E},\;(i,j) \notin v(G)\label{phiij}\\
\Phi_{ii} &=|\Phi^2_{i_{u}i_{u}}+\sum\limits^{i_{u}-1}_{k=1}\Phi^2_{ki_{u}}-\sum\limits^{i-1}_{k=1}\Phi^{2}_{ki}|^{\frac{1}{2}}, &\mbox{for}\;i=1,\ldots,p.\label{phiii}
\end{align}
where $(i_{u},j_{u})=\min\{(i,j): i\leq j$ and $(i,j)\in u \in \mathcal{V}\cup \mathcal{E} \}$ in the lexicographical order.
\end{proposition}
\begin{proof}
The first three equations can be found in \citet{Rov:2002}. We will only prove \eqref{phiij} since \eqref{phiii}  will follow immediately from  it.
For all $(i,j)\in u \in \mathcal{V}\cup \mathcal{E}$ and $(i,j)\neq (i_{u}, j_{u})\in u$, by \eqref{1}, we have that
$K_{i_{u}j_{u}}=\sum\limits^{i_{u}}_{k=1}\Phi_{ki_{u}}\Phi_{kj_{u}}$
and in general
$K_{ij}=\sum\limits^{i}_{k=1}\Phi_{ki}\Phi_{kj}.$
Since $K_{ij}=K_{i_{u}j_{u}}$, it follows that
$$\Phi_{i_{u}j_{u}}\Phi_{i_{u}i_{u}}+\sum\limits^{i_{u}-1}_{k=1}\Phi_{ki_{u}}\Phi_{kj_{u}}=\Phi_{ii}\Phi_{ij}+\sum\limits^{i-1}_{k=1}\Phi_{ki}\Phi_{kj}.$$
Equations \eqref{phiij} and \eqref{phiii} follow then immediately.
\end{proof}
\begin{proposition}
\label{chokpsi}
For $K=Q^T(\Psi^T\Psi)Q\in P_{\cal G}$ with $\Psi$ and $Q$ as defined above, the entries $\Psi_{ij}$ of $\Psi$ are as follows:
\begin{align}
\Psi_{rs} &=\sum\limits_{j=r}^{s-1}-\Psi_{rj}\frac{Q_{js}}{Q_{ss}}+\frac{\Phi_{rs}}{Q_{ss}}\;\;\;\;\mbox{for}\; (r,s)\in v(G),\;r\neq s\;,\label{2}\\
\Psi_{ss}&=\frac{\Phi_{ss}}{Q_{ss}}\;\;\;\;\mbox{for}\;(r,s)\in v(G),\;r=s\;,\nonumber\\
\Psi_{rs}&=\sum\limits_{j=r}^{s-1}-\Psi_{rj}\frac{Q_{js}}{Q_{ss}}-\sum\limits^{r-1}_{i=1}(\frac{\Psi_{ir}+\sum\limits^{r-1}_{j=i}\Psi_{ij}\frac{Q_{jr}}{Q_{rr}}}{\Psi_{rr}})(\Psi_{is}+\sum\limits^{s-1}_{j=i}\Psi_{ij}\frac{Q_{js}}{Q_{ss}})\;\;\;\;\mbox{for}\;K_{rs}=0,\;r\neq 1\;,\nonumber\\
\Psi_{1s}&=\sum\limits_{j=1}^{s-1}(-\Psi_{1j}\frac{Q_{js}}{Q_{ss}})\;\;\;\;\mbox{for}\;K_{1s}=0\;,\nonumber\\
\Psi_{rs}&=\frac{\Phi_{i_{u}j_{u}}\Phi_{i_{u}i_{u}}+\sum\limits^{i_{u}-1}_{k=1}\Phi_{ki_{u}}\Phi_{kj_{u}}-\sum\limits^{r-1}_{k=1}\Phi_{kr}\Phi_{ks}}{\Phi_{rr}Q_{ss}}-\sum\limits^{s-1}_{j=r}\Psi_{rj}\frac{Q_{js}}{Q_{ss}}\;\;\;\;\mbox{for}\;K_{rs}\neq 0,\label{psirs}\\
&\; \hspace{9cm}(r,s) \in u \in \mathcal{V} \cup \mathcal{E}, (r,s)\notin v(G)\;,\nonumber\\
\Psi_{ss}&=\frac{|\Phi^2_{i_{u}i_{u}}+\sum\limits^{i_{u}-1}_{k=1}\Phi^2_{ki_{u}}-\sum\limits^{r-1}_{k=1}\Phi^{2}_{ks}|^{\frac{1}{2}}}{Q_{ss}}\;\;\;\;\mbox{for}\;s=1,\ldots,p\;.\label{psiss}
\end{align}
 \end{proposition}
\begin{proof}
 We will prove \eqref{psirs} and therefore \eqref{psiss}. Since $\Phi=\Psi Q$, for $r \neq s$, we have
 $$\Phi_{rs}=\Psi_{rs}Q_{ss}+\sum\limits^{s-1}_{j=r}\Psi_{rj}Q_{js}.$$
On the other hand, by \eqref{2}, we have
$$\Phi_{rs}=\frac{\Phi_{i_{u}j_{u}}\Phi_{i_{u}i_{u}}+\sum\limits^{i_{u}-1}_{k=1}\Phi_{ki_{u}}\Phi_{kj_{u}}-\sum\limits^{r-1}_{k=1}\Phi_{kr}\Phi_{ks}}{\Phi_{rr}}.$$
It then follows that

$$\Psi_{rs}Q_{ss}+\sum\limits^{s-1}_{j=r}\Psi_{rj}Q_{js}=\frac{\Phi_{i_{u}j_{u}}\Phi_{i_{u}i_{u}}+\sum\limits^{i_{u}-1}_{k=1}\Phi_{ki_{u}}\Phi_{kj_{u}}-\sum\limits^{r-1}_{k=1}\Phi_{kr}\Phi_{ks}}{\Phi_{rr}}$$
which implies \eqref{psirs} and \eqref{psiss}.
\end{proof}

Next we compute the Jacobian of the change of variable from $K$ to $\psi^{v(G)}$ in two steps.

\begin{lemma}
\label{jac1}
Let $K$ be in $P_{\cal G}$. Let $v_i^G$ be the number $j\in \{i,\ldots,p\}$ such that $(i,j)\not\in v(G)$. Then the Jacobian of the change of variable $K^{v(G)}\rightarrow \Phi ^{v(G)}$  as defined above is
$$\det (J(K^{v(G)}\rightarrow \Phi^{v(G)}))=2^{|\mathcal{V}|}\prod\limits^{p}_{i=1}\Phi_{ii}^{p-i+1-v_i^G}$$
where  $|\mathcal{V}|$ is the number of vertex color class of $\G$.
\end{lemma}
\begin{proof}
Order the elements of both matrices $K$ and $\Phi$ according to the lexicographic order.
For $(i,j)\in v(G)$, differentiating \eqref{1} yields
\begin{align}
&\frac{\partial K_{ii}}{\partial \Phi_{ii}}=2\Phi_{ii}\;,&\frac{\partial K_{ii}}{\partial \Phi_{ks}}=0\;\mbox{for}\; (k,s)>(i,i)\;,\\
&\frac{\partial K_{ij}}{\partial \Phi_{ij}}=\Phi_{ii}\;,& \frac{\partial K_{ij}}{\partial \Phi_{ks}}=0\;\mbox{for}\;(k,s)>(i,j),\;\;i\neq j.
\end{align}
Therefore, the Jacobian is an upper-triangular matrix and its determinant is the product of the diagonal elements.
The lemma then follows immediately from the fact that  for $i\in \{1,\ldots,p\}$ given, the cardinality of the set $\{(i,j)\in v(G),\; (i,j)\geq (i,i)\}$ is $p-i+1-v_i^G$.
\end{proof}

\begin{lemma}
\label{jac2}
Let  $K$ be in $P_{\cal G}$. Let  $d_{i}^{G}=\mid\{j:j\leq i,\;(j,i)\notin v(G)\}\mid$.
The Jacobian of the change of variable $\Phi^{v(G)}\rightarrow\Psi^{v(G)}$ where $\Phi$ and $\Psi$ are as defined above  is
$$\det (J(\Phi^{v(G)} \rightarrow \Psi^{v(G)}))=\prod\limits^{p}_{i=1}Q_{ii}^{i-d_{i}^{G}}.$$
\end{lemma}
\begin{proof}
Order the elements of both matrices $\Phi$ and $\Psi$ according to the lexicographic
order. For $(r,s) \in v(G)$, differentiating \eqref{2}, we obtain
\begin{align}
&\frac{\partial \Phi_{rs}}{\partial \Psi_{ss}}=Q_{ss}\;,&\frac{\partial \Phi_{rs}}{\partial \Psi_{ij}}=0\;\mbox{for}\;(i,j)>(r,s).
\end{align}
The Jacobian is thus an upper-triangular matrix and its determinant is the product of the diagonal elements. The lemma follows from the definition of $d_{i}^{G}.$
\end{proof}

\begin{theorem}
\label{density}
Let $\mathcal{G}=(\mathcal{V}, \mathcal{E})$ be an arbitrary $p$-dimensional colored graph. Then the density of the $CG$-Wishart distribution expressed in terms of $\Psi^{v(G)}$ is
\begin{eqnarray}
p(\Psi^{v(G)}|\delta,D) =\frac{2^{|\mathcal{V}|}}{I_{G}(\delta,D)}\prod\limits_{i=1}^{p}Q_{ii}^{p-v_{i}^{G}-d_{i}^{G}+\delta-1}\prod\limits_{i=1}^{p}\Psi_{ii}^{p-i-v_{i}^{G}+\delta-1}
e^{-\frac{1}{2}\sum\limits_{i=1}^{p}\sum\limits^{p}_{j=i}\Psi_{ij}^{2}}.\label{dens}
\end{eqnarray}
\end{theorem}
\begin{proof}
The expression of  $p(\Phi^{v(G)}|\delta,D)$  above  follows immediately from the fact that $|K|=\prod\limits^{p}_{i=1}\Phi_{ii}^{2}$, that $\langle K,D\rangle =\sum_{i=1}^p\sum_{j=1}^p\Psi_{ij}^2$ and from the expressions of the Jacobians given in Lemmas \ref{jac1} and \ref{jac2}.
\end{proof}
\subsection{The sampling algorithm}
We now briefly describe the MH algorithm we use to sample from the density \eqref{dens}. We first note that if we make the further change of variables
$$(\Psi_{ii}, (i,i)\in v(G), \Psi_{ij}, (i,j)\in v(G), i\neq j)\mapsto (t_{ii}=\Psi_{ii}^2, (i,i)\in v(G), \Psi_{ij}, (i,j)\in v(G), i\neq j)\;,$$
we obtain
\begin{eqnarray*}
p(t_{ii}, (i,i)\in v(G), \Psi_{ij}, (i,j)\in v(G), i\neq j|\delta,D)&\propto &\prod\limits_{i=1}^{p}t_{ii}^{\frac{p-i-v_{i}^{G}+\delta}{2}-1}e^{-\frac{1}{2}\sum\limits_{i=1}t_{ii}}e^{-\frac{1}{2}\sum\limits^p_{i=1}\sum\limits^{p}_{j=i+1}\Psi_{ij}^{2}}
\end{eqnarray*}
and we observe that $t_{ii}^{\frac{p-i-v_{i}^{G}+\delta}{2}-1}e^{-\frac{1}{2}\sum\limits_{i=1}t_{ii}}$ has the form of a $\chi^2_{p-i-v_{i}^{G}+\delta}$ distribution.

We denote by $\Psi^{[s]}$ and $\Psi^{[s+1]}$ the current state of the chain and the next state of the chain, respectively. We denote $\Psi'$ the candidate for $\Psi^{[s+1]}$. We also use the  notation
$$\Psi_{v(G)^c}=
\Big(\Psi_{ij}, (i,j)\in v(G)^c\Big) $$
where $v(G)^c$ is the complement of $v(G)$ in $V\times V$.
For $(i,j) \in v(G)$, an element $\Psi_{ij}^{[s]}$ is updated by sampling a value $\Psi'_{ij}$ from a normal distribution with zero mean and standard deviation equal to one. For $(i,i) \in v(G)$, a element $\Psi_{ii}^{[s]}$ is updated by sampling a value $(\Psi^2_{ii})'$ from a chi-square distribution with
$p-i-v_{i}^{G}+\delta$ degrees of freedom. The non-free elements of $\Psi'$ are uniquely defined by the functions in Proposition \ref{chokphi} and Proposition \ref{chokpsi}. The Markov chain moves to $\Psi'$ with probability $$\min\{\frac{h[(\Psi')_{v(G)^c}]}{h[(\Psi^{[s]})_{v(G)^c}]},1\},$$
where
$$h(\Psi_{v(G)^c}) =\prod\limits_{(i,i) \in v(G)^c}\Psi_{ii}^{p-i-v_{i}^{G}+\delta-1} \exp(-\frac{1}{2}\sum\limits_{(i,j) \in v(G)^c}\Psi_{ij}^{2}).$$
Finally, we can obtain $K^{[s]}=Q^T(\Psi^{[s]})^T\Psi^{[s]}Q$. Since $h(\Psi_{v(G)^c})$ is uniformly bounded by 1, the chain is uniformly ergodic (the strongest convergence rate in use, see \cite{MengTweed:1996}).

We now have a method to sample values of $K$ from the $CG$-distribution, whether it is as a prior or a posterior distribution and thus obtain an estimate of the posterior mean of $K$. In our MH algorithm, the candidates are drawn independently of the current samples through the proposal density. Thus, the algorithm gives an independence MH chain. Our simulation results in Section 5 will show that the chain has good mixing,  low autocorrelation and high proximity to the true $CG$-distribution.

The  sample mean will converge to the expected value of $K$. In order to verify the accuracy of our sampling algorithm, we therefore would like to have the exact value of the expected value of $K$ under the $CG$-Wishart.
This is done in the next section for some special coloured graphs.
\section{The exact expected value of $K$ in some special cases}
\subsection{The  mean of the $CG$-Wishart}
For a given ${\cal G}$, the $CG$-Wishart as defined in \eqref{cg} and \eqref{norm}  clearly form a natural exponential family of the type
$$f(K;\theta)d K=\exp \{\langle K, \theta\rangle -k(\theta)\}\mu(d K)$$
with generating measure $\mu(d K)=|K|^{(\delta-2)/2}{\bf 1}_{P_{\cal G}}(K)$, $\theta=-\frac{1}{2} D$ and cumulant generating function  $k(\delta,D)=\log I_{\cal G}(\delta,D).$ To verify the accuracy of the sampling method given in Section 3, we will compare the expected value of $K$ under the $CG$-Wishart and the  sample mean obtained from a number of iterations of our MH algorithm. From the theory of natural exponential family, we know that the mean of the $CG$-Wishart is
$$E(K)=\frac{\partial \; k(\delta,D)}{\partial \; (-\frac{1}{2} D)}=-2\frac{\partial \; k(\delta,D)}{\partial D}.$$
We therefore need to determine for which values of $\delta$ and $D$ the quantity $ I_{\cal G}(\delta,D)$ is finite and then compute the analytic expression of $ I_{\cal G}(\delta,D)$. We need also to differentiate
this expression.

 We cannot do this in general but we will now consider  several particular coloured graphs for which we can compute $I_{\G}(\delta,D)$. For the corresponding RCON models, we will see that when $\delta>0$ (except in the case of the star graph with all leaves in the same colour class where we must have $\delta\geq 1$), the normalizing constant $ I_{\cal G}(\delta,D)$ is finite when $D$ belongs to the dual  $P^*_{\G}$ of $P_{\G}$. For any open convex cone $C$ in  $R^n$, the dual of $C$ is defined as
$$C^*=\{y\in R^n\mid \langle x, y\rangle >0,\;\forall x\in \bar{C}\setminus \{0\}\}$$
where $\bar{C}$ denotes the closure of $C$.

In the remainder of this section,  for each RCON model,  we determine $P_{\G}^*$ and the value of $ I_{\cal G}(\delta,D)$.
  This will allow us, in Section 5, to verify the accuracy of our sampling method.

  All proofs for Section 4 are given in Appendix 1 in the Supplementary file.
\subsection{ Trees with vertices of different colours and edges of the same colour}
Let $V=\{1,\ldots,p\}$. Let $T=(V,E)$ be a tree with vertices of different colours and edges of the same colour. An example of such $\G$ is given  in Figure \ref{fig:1}(a). Let $a=(a_i,i=1,\ldots,p)^t$ where $a_{i}\geq 0$ and $b\in R$.
 Let $S$ be the space of symmetric $p\times p$ matrices.  We define the mapping
 \begin{equation}
 \label{map}
 m: \;(a,b)\in R^{p+1}\mapsto m(a,b)\in S\; \end{equation}
with $m(a,b)$ satisfying the conditions $$ [m(a,b)]_{ii}=a_i, \;[m(a,b)]_{ij}=b=[m(a,b)]_{ji}\;\mbox{for}\;(i,j)\in E,\;[m(a,b)]_{ij}=0 \;\mbox{for}\;(i,j)\not \in E.$$

 Let $M(\G)$ be the linear space of matrices $m(a,b)$ for $(a,b)\in R^{p+1}$.
Let $P$ be the cone of $p\times p$ symmetric positive definite matrices. Then
\begin{equation}
\label{pgt}
P_{\cal G}=M(\G)\cap P\;.
\end{equation}
\begin{proposition}
\label{pgtree}
Let $T$ be a tree as described above. The dual cone $P^*_{\cal G}$ is
\begin{equation}
\label{pgtgiven}
P^*_{\cal G}=\{m(a',b')\in M(\G)\mid a'=(a'_i, i=1,\ldots,p),\;b'\in R, \;|b'|<\frac{1}{p-1}\sum_{(i,j)\in E}\sqrt{a_i'a_j'}\}\;.
\end{equation}
\end{proposition}
We are now in a position to give the analytic expression of $I_{\cal G}(\delta,D)$.
\begin{theorem}
\label{pgtnorm:theorem}
For $\G=T$ as described above, $\delta>0$  and $D=m(a',b')\in P^*_{\cal G}$, the normalizing constant  $I_{\cal G}(\delta, D)$ is finite and equal to
\begin{equation}
\label{pgtnorm}
I_{\G}(\delta, D)=2^{\frac{\delta}{2}+p-1}\Gamma(\frac{\delta}{2})\left(\prod_{i=1}^p(a'_i)^{d_i-2}\right)^{\frac{\delta}{4}}\int_{-\infty}^{\infty}\left(\prod_{(i,j)\in E}K_{\frac{\delta}{2}}(|b|\sqrt{a'_ia'_j})\right)|b|^{\frac{p\delta}{2}}e^{-(p-1)bb'}db
\end{equation}
where  $d_i$ denotes the number of neighbours of the vertex $i$ in the tree ($V,E)$.
\newline For $\delta=1$, we have
$$I_{\G}(1,D)=(2\pi)^{\frac{p}{2}}\prod_{i=1}^p(a'_{i})^{-\frac{1}{2}}\Big([\sum_{(i,j)\in E}(a'_{i}a'_{j})-(p-1)b']^{-1}+[\sum_{(i,j)\in E}(a'_{i}a'_{j})+(p-1)b']^{-1}\Big).$$
For $\delta=3$,  Let  $\sigma_k$ the $k$th elementary function of the variables $\sqrt{a'_ia'_{j}},\;(i,j)\in E.$
We have
\begin{eqnarray*}
I_{\G}(3,D)&=&2^{\frac{p}{2}-1}\pi^{\frac{p}{2}}\prod_{i=1}^{p}(a'_{i})^{-\frac{3}{2}}\sum\limits^{p-1}_{k=0}\sigma_{k}\Gamma(k+1)
\big(
\big[\sum_{(i,j)\in E}(a'_{i}a'_{j})^{\frac{1}{2}}-(p-1)b'\big]^{-(k+1)}\\
&&\hspace{6cm}-\big[\sum_{(i,j)\in E}(a'_{i}a'_{j})^{\frac{1}{2}}+(p-1)b'\big]^{-(k+1)}
\big).
\end{eqnarray*}
\end{theorem}
\subsection{The star graph with its $n$ leaves in one colour class}
An example of star graph with its $n$ leaves in one color class and different colors for the edges and the central node is given in Figure \ref{fig:1}(b).
For $a\in R, c\in R, b=(b_1,\ldots,b_n)\in R^n$, let $L(\G)$ the linear space of matrices of the form
$$l(a,b,c)=\left[\begin{array}{ccccc}a&b_1&b_2&\ldots&b_n\\b_1&c&0&\ldots&0\\b_2&0&c&\ldots&0\\\ldots&\ldots&\ldots&\ldots&\ldots\\b_n&0&0&\ldots&c\end{array}\right].$$
It is easy to see that the determinant of $l(a,b,c)$ is
\begin{equation}\label{DK} | l(a,b,c)|=c^n\left(a-\frac{\|b\|^2}{c}\right)\end{equation}
and therefore,  $P_{\cal G}$ is the open cone
$$P_{\cal G}=\{l(a,b,c)\in L({\cal G}):\; c>0,\;a-\frac{\|b\|^2}{c}>0.\}$$
The dual cone $P^*({\cal G})$ and the normalizing constant $I_{\cal G}(\delta, D)$ are given below.
\begin{proposition}
\label{pgdaisyleaves}
For a star graph with all $n$ leaves in one colour class, the dual of $P_{\cal G}$ is
\begin{equation}
P_{\cal G}^*=\{l(a',b',c')\in L({\cal G})\mid ||b'||^2\leq na'c'\}.
\end{equation}
\end{proposition}
\begin{theorem}
\label{igdaisyleaves}
For $\G$ a star graph with all $n$ leaves in the same colour class, $\delta\geq 1$ and $D=l(a',b',c')\in P_{\G}^*$, the normalizing constant of the $CG$-Wishart is
$$I_{\cal G}(\delta, D)=2^{\frac{\delta+n\delta+2}{2}}\pi^{n/2}\times a'^{(\frac{\delta}{2}-1)(n-1)}\times \frac{1}{(na'c'-\|b'\|^2)^{(\delta-1)\frac{n}{2}+1}}\times \Gamma((\delta-1)\frac{n}{2}+1)\Gamma(\frac{\delta}{2}).$$
\end{theorem}
\subsection{The star graph with all vertices in one colour class}
An example of star graph with all vertices in one color class and different colors for the edges  is given in Figure \ref{fig:1}(c).
This case is a special case of the preceding one and therefore, we have immediately that
$$P_{\cal G}=\{l(a,b,a)\in L({\cal G}) \mid a>0, a^2-||b||^2>0\}.$$
Since this is a well-known cone, called the Lorentz cone, we know also that it is self dual and therefore
$$P^*_{\cal G}=\{l(a',b',a')\in L({\cal G}) \mid a'>0, (a')^2-||b'||^2>0\}.$$
It remains to compute $I_{\cal G}(\delta, D)$.
\begin{theorem}
\label{igdaisyall}
For ${\cal G}$ the star graph with $n$ leaves and all vertices in the same colour class, $\delta>0$ and $D=l(a',b',a')\in P_{\G}^*$, the normalizing constant of the $CG$-Wishart is
\begin{eqnarray*}
I_{\G}(\delta, D)&=&\frac{2^{\frac{(n+1)\delta}{2}-1}C_n\Gamma((n+1)\frac{\delta}{2})}{(n+1)^{\frac{(n+1)\delta}{2}}(a')^{\frac{(n+1)\delta}{2}}}B(\frac{\delta}{2},\frac{n}{2}) \;_2F_1\Big((n+1)\frac{\delta}{4}, (n+1)\frac{\delta}{4}+\frac{1}{2},\frac{n+\delta}{2};u\Big)
\end{eqnarray*}
where $u=\Big(\frac{2||b'||}{(n+1)a'}\Big)^2$ and $B(\frac{\delta}{2},\frac{n}{2})$ is the Beta function with argument $(\frac{\delta}{2},\frac{n}{2})$.
\end{theorem}

\begin{figure}
         \centering
         \begin{subfigure}[b]{0.2\textwidth}
                 \includegraphics[width=\textwidth]{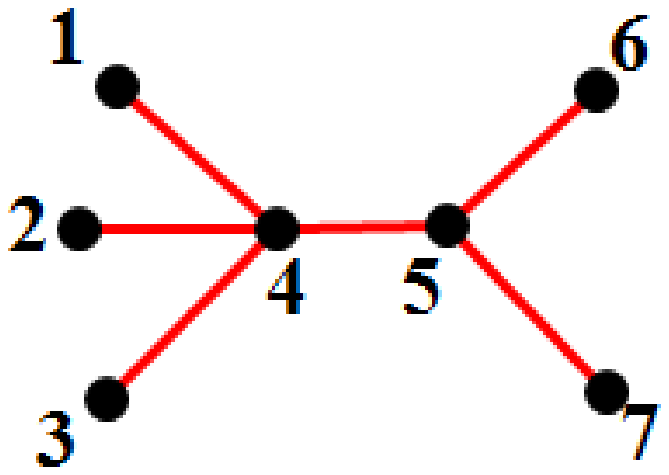}
                 \caption{}
                 \label{fig:tree}
         \end{subfigure}%
         ~ 
         \begin{subfigure}[b]{0.15\textwidth}
                 \includegraphics[width=\textwidth]{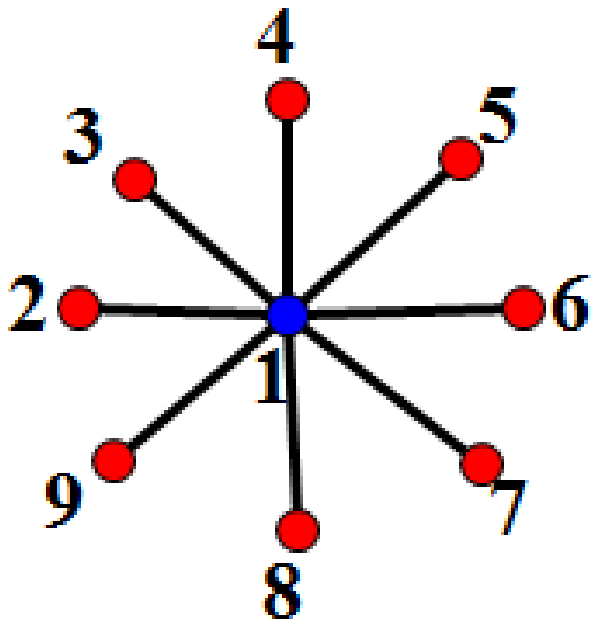}
                 \caption{}
                 \label{fig:star}
         \end{subfigure}
         ~ 
         \begin{subfigure}[b]{0.15\textwidth}
                 \includegraphics[width=\textwidth]{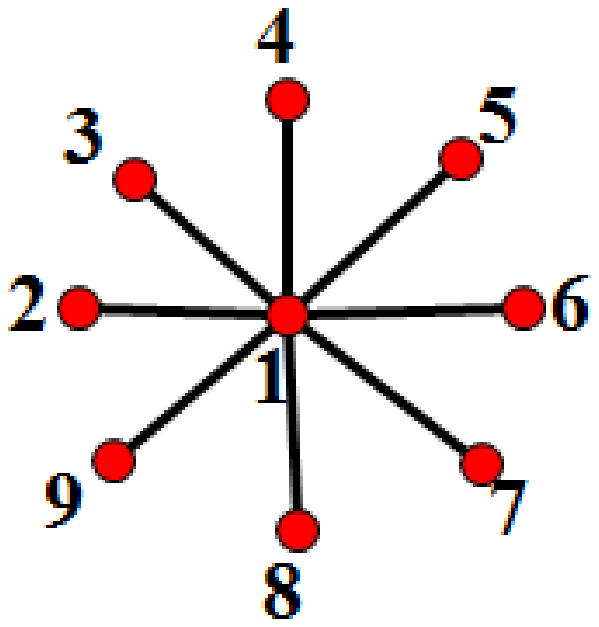}
                 \caption{}
                 \label{fig:star1}
         \end{subfigure}
         ~
         \begin{subfigure}[b]{0.15\textwidth}
                 \includegraphics[width=\textwidth]{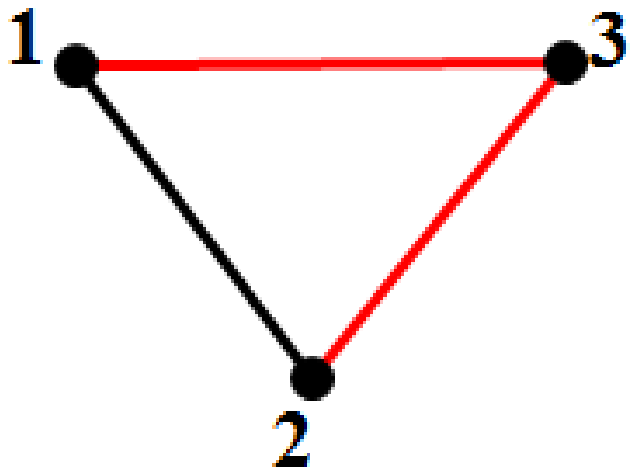}
                 \caption{}
                 \label{fig:triangle}
         \end{subfigure}
         ~
                  \begin{subfigure}[b]{0.15\textwidth}
                 \includegraphics[width=\textwidth]{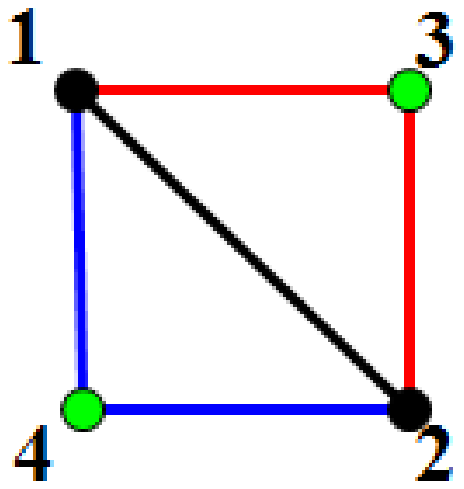}
                 \caption{}
                 \label{fig:decomposable}
         \end{subfigure}

         \caption{Black vertices and edges are all of different colours. (a) The colored tree. (b) The coloured star with the centre vertex of a different colour. (c) The coloured star with all vertices of the same colour. (d) The triangle with two edges of the same colour. (e) The decomposable graph with three different colours for the edges.}
         \label{fig:1}
\end{figure}

\subsection{A complete graph on three vertices with two edges in the same colour class}
This graph is represented in Figure \ref{fig:1}(d).
In this case, the cone $P_{\cal G}$ is the set of positive definite matrices $K=(k_{ij})_{1\leq i,j\leq 3}$ with $k_{13}=k_{23}$.

The dual cone $P^*({\cal G})$ and the normalizing constant $I_{\cal G}(\delta, D)$ are given below.
\begin{proposition}
\label{pgtriangle}
For the graph in Figure \ref{fig:1}(d), the dual of $P_{\G}$ is
\begin{eqnarray*}
P^*_{\cal G}&=&\{ D=(d_{ij})_{1\leq i,j\leq 3}\in S\mid d_{13}=d_{23},\\
&&\hspace{2cm}\;d_{ii}>0,i=1,2,3,\;d_{12}^2<d_{11}d_{22},\;4d_{13}^2<(d_{11}+d_{22}+2d_{12})d_{33}
\}.
\end{eqnarray*}
\end{proposition}

\begin{theorem}
\label{igtriangle}
For $\G$ as in Figure \ref{fig:1}(d),  $\delta>0$ and $D\in P^*_{\cal G}$, the normalizing constant of the $CG$-Wishart is
\begin{eqnarray*}
I_{\cal G}(\delta, D)
 &=& 2^{\frac{3\delta+4}{2}}\pi\Gamma(\frac{\delta}{2})\Big(\Gamma(\frac{\delta+1}{2})\Big)^2(d_{11}+d_{22}+2d_{12})^{\frac{\delta}{2}}[d_{33}(d_{11}+d_{22}+2d_{12})-4d^2_{13}]^{-\frac{\delta+1}{2}}\\
&&\hspace{2cm}\; \times (d_{11}d_{22}-d^2_{12})^{-\frac{\delta+1}{2}}.
\end{eqnarray*}
\end{theorem}
\subsection{A decomposable graph with three vertex classes and three edge classes}
This graph is represented in Figure \ref{fig:1}(e). Then the cone $P_{\cal G}$ is the set of matrices of the form
$$K=\left(\begin{array}{cccc}k_{11}&k_{12}&k_{13}&k_{14}\\k_{12}&k_{22}&k_{13}&k_{14}\\
k_{13}&k_{13}&k_{33}&0\\k_{14}&k_{14}&0&k_{33}\end{array}\right).$$
\begin{proposition}
\label{pgdecomp}
For ${\cal G}$ as in Figure \ref{fig:1}(e), the dual cone is the set of matrices
\begin{eqnarray*}
P_{\cal G}^*&=&\{D=(d_{ij})_{1\leq i,j\leq 4}\in S\mid d_{23}=d_{13},d_{24}=d_{14},d_{44}=d_{33},\;d_{11}>0,\;\\
&&d_{11}d_{22}-d_{12}^2>0,\;d_{11}+2d_{12}+d_{22}>0,\;d_{33}(d_{11}+2d_{12}+d_{22})-2(d_{13}^2+d_{14}^2)>0\}\;.
\end{eqnarray*}
\end{proposition}
\begin{theorem}
\label{igdecomp}
For ${\cal G}$ as in Figure \ref{fig:1}(e),  $\delta>0$ and $D\in P_{\G}^*$, the normalizing constant of the $CG$-Wishart is
\begin{eqnarray*}
I_{\cal G}(\delta, D)&=& 2^{\delta+2}\pi^{\frac{3}{2}}\Gamma(\frac{\delta}{2})\Gamma(\frac{\delta+1}{2})\Gamma(\delta)(d_{11}+d_{22}+2d_{12})^{\delta-1}(d_{11}d_{22}-d^2_{12})^{-\frac{\delta+1}{2}}\\
&&\hspace{1cm}\times [d_{33}(d_{11}+d_{22}+2d_{12})-2(d^2_{13}+d^2_{14})]^{-\delta}.
\end{eqnarray*}
\end{theorem}
\section{Numerical experiments when we know the exact mean}
In order to illustrate the performance of our MH algorithm, we conduct a numerical experiment for each of the colored graph (a) - (e) shown in Figure \ref{fig:1}. In each case,  for a given $D$ and $\delta$, we first derive $\log I_{G}(\delta, D)$, then the prior mean $E(K)$ under the $CG$-Wishart by differentiating $\log I_{G}(\delta, D)$ with respect to $-\frac{D}{2}$. We then sample from the $CG$-Wishart. We run the chain for 5000 iterations and discard the first 1000 samples as burn in. Our estimate $\hat{K}$ for $K$ is the average $\hat{K}=\frac{\sum_{i=1001}^{i=5000}\hat{K}_{i}}{4000}$ of the remaining 4000 iterations $\hat{K}_{i}, i=1001,\ldots, 5000$.  For arbitrary $K$ and $K'$ we define the normalized mean square error ($nmse$) between $K$ and $K'$ to be
$$nmse(K,K')=\frac{||K-K'||_2^2}{||K'||^2_2}$$
where $||K||_2^2$ is the sum of the squares of the entries of $K$.
We repeat the previous experiment 100 times, obtain  $\hat{K}^j,j=1,\ldots,100$ and compute
$$\overline{nmse}(E(K),\hat{K})=\frac{1}{100}\sum_{j=1}^{100}nmse(E(K),\hat{K}^j)$$
where $E(K)$ is obtained by differentiation of $\log I_{\cal G}(\delta, D)$ with respect to $-\frac{D}{2}$ at  our given $D$ and $\delta$.

For each graph in Figure \ref{fig:1}, for an arbitrary $j\in \{1,\ldots,100\}$, we give the trace plot of $\log |K^j_i|, \; i=1000,\ldots,5000$. The traceplot shows that the chain seems to be mixing well. We also provide the autocorrelation plot with time-lag $h$ for $\log |K^j_i|, \; i=1000,\ldots,5000$ in function of $h$ where, for an arbitrary given $j$, we define the autocorrelation coefficient for $Y_i=\log |K^j_i|, \; i=1000,\ldots,5000$ to be
$$R_h=\frac{\sum_{i=1000}^{5000-h}\;(Y_i-\bar{Y})(Y_{i+h}-\bar{Y})}{\sum_{i=1000}^{5000}\;(Y_i-\bar{Y})^2}.$$
 The autocorrelation plots indicate that the samples have  a low autocorrelation.
 The numerical values of the matrices $D$, $E(K)$ and $\hat{K}$ as well as the traceplot and autocorrelation plot of $\log (|K|)$ for all five graphs in Figure \ref{fig:1} are given in Appendix 2 in the Supplementary file.
  An overview of calculations and results are given in
Table \ref{table:1} which,  for all different five colored graphs in Figure \ref{fig:1}, shows the parameter $\delta$ we chose for the prior distribution, $\log I_{\cal G}(\delta,D)$ and the normalized mean square errors. In order to obtain the mean $E(K)$ of the $CG$-Wishart for the graph in Figure \ref{fig:1}(c), we use formula \eqref{dhy} to get the derivative of the hypergeometric function $_{p}F_{q}(a_{1},\ldots, a_{p};b_{1},\ldots,b_{q};z)$. We see that the normalized mean square error is of the order of $10^{-3}$ or less except for
the star graph with all leaves of the same colour in Figure \ref{fig:1}(b).

\begin{table}
\begin{tabular}{|c|c|c|c|}
\hline
$\G$ &$\delta$ &$\log I_{\cal G}(\delta,D)$& $\overline{nmse}(E(K),\hat{K})$\\
\hline
Fig. \ref{fig:1}(a)& 1 &$-\frac{1}{2}\sum\limits_{i=1}^{7}\log a'_{i}+\log\left[\frac{1}{\sum\limits^{6}_{i=1}(a'_{i}a'_{j_{i}})^{\frac{1}{2}}-6b'}-\frac{1}{\sum\limits^{6}_{i=1}(a'_{i}a'_{j_{i}})^{\frac{1}{2}}+6b'}\right]$&0.0069\\
\hline
&&&\\
Fig. \ref{fig:1}(b)&3&$\frac{7}{2}\log a'-9\log(8a'c'-\|b'\|^2)$&0.0187\\
\hline
&&&\\
Fig. \ref{fig:1}(c)&3&$-15\log a'+\log\; _2F_1\left(\frac{15}{2},8;6;\frac{\|b'\|^2}{25a'^2}\right)$&0.0064\\
\hline
&&&\\
Fig. \ref{fig:1}(d)&3&$\frac{3}{2}\log d-2\log(d_{33}d-4d^2_{13})-2\log(d_{11}d_{22}-d^2_{12})$&0.0005\\
\hline
&&&\\
Fig. \ref{fig:1}(e)&3&$2\log d-3 \log(d_{33}d-2d^2_{13}-2d^2_{14})-2\log(d_{11}d_{22}-d^2_{12})$&0.0009\\
&&&\\
\hline
\end{tabular}
\caption{For the graphs of Fig. \ref{fig:1} and $\delta$ given: analytic expression of $\log I_{\cal G}(\delta,D)$ where $d=d_{11}+d_{22}+2d_{12}$ and value of $\overline{nmse}(E(K), \hat{K})$  averaged over  100 experiments. }
\label{table:1}
\end{table}

\vspace{2mm}



\section{The posterior mean from simulated data: $p=20, \;p=30$}
In this section, in order to assess the accuracy of our sampling method for larger graphs, we  generate data from a $N(0, K^{-1})$ distribution with $K$ given in $P_{\G}$. We take the $CG$-Wishart with $\delta=3$ and $D=I$ as the prior distribution of $K$. Clearly the posterior distribution will be $CG$-Wishart with parameters $\delta+n$ and $I+nS$ where $S$ is the sample covariance matrix. We will use this posterior and our sampling method of Section 3 to compute the posterior mean  $E(K|S)$ as an estimate of $K$.



\begin{figure}
         \centering
         \begin{subfigure}[b]{0.3\textwidth}
                 \includegraphics[width=\textwidth]{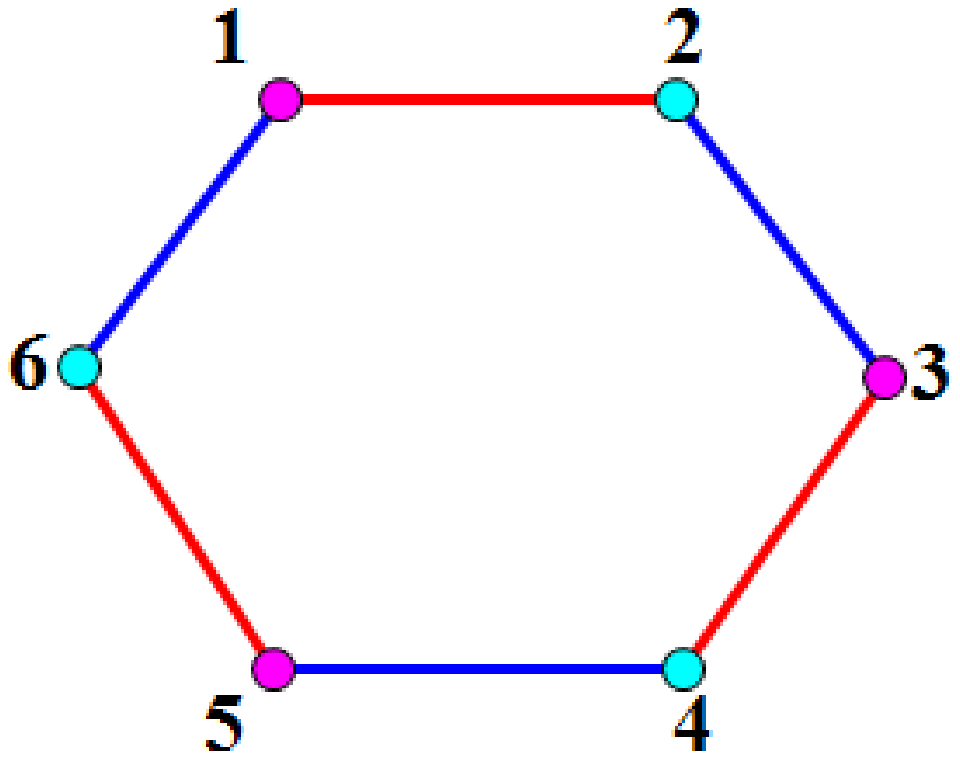}
                 \caption{}
                 \label{fig:vertex and edge}
         \end{subfigure}%
         ~ 
         \begin{subfigure}[b]{0.3\textwidth}
                 \includegraphics[width=\textwidth]{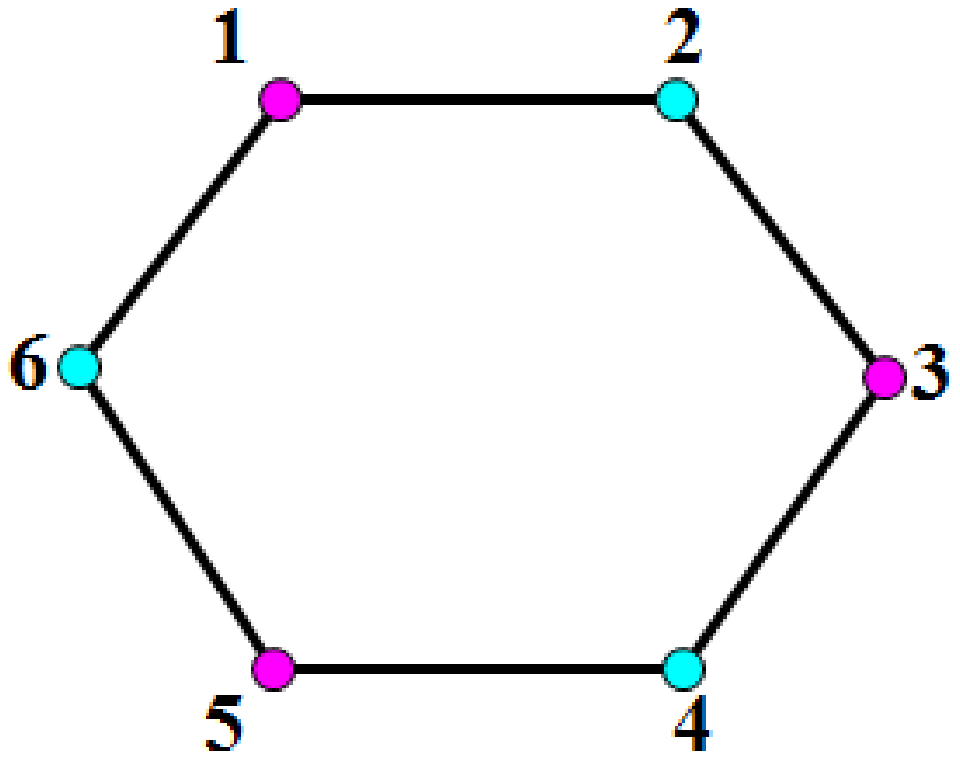}
                 \caption{}
                 \label{fig:vertex}
         \end{subfigure}
         ~ 
         \begin{subfigure}[b]{0.3\textwidth}
                 \includegraphics[width=\textwidth]{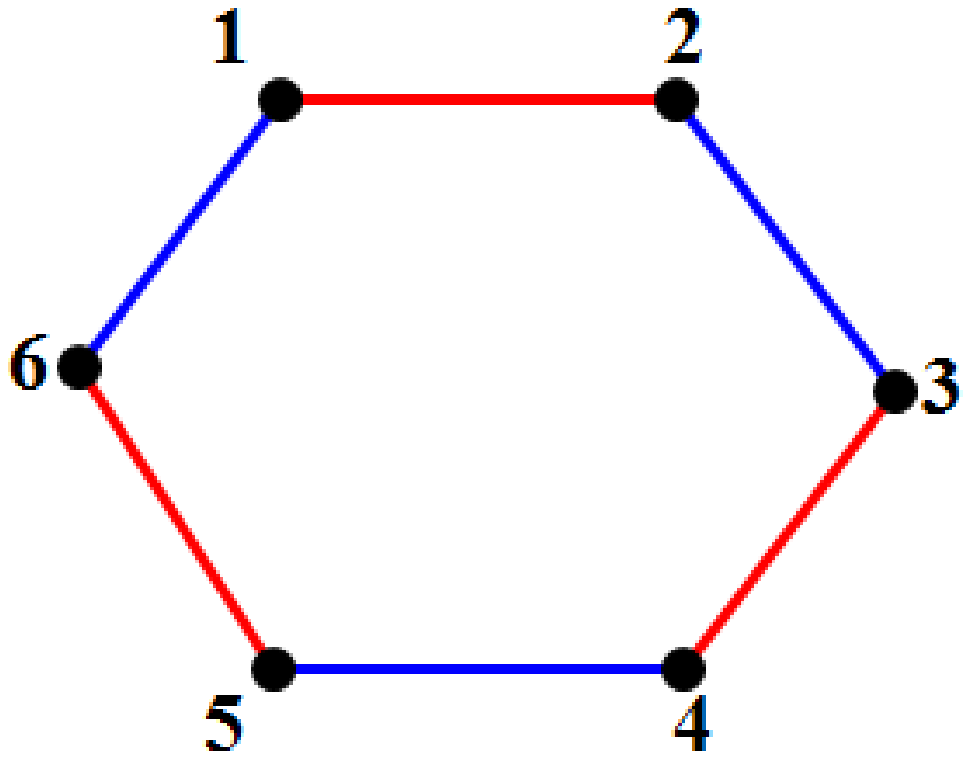}
                 \caption{}
                 \label{fig:edge}
         \end{subfigure}

         \caption{Cycles of length 6 with the three different patterns of colouring that we use for the cycles of length $p=20$ and $p=30$. Black vertices or edges indicate different arbitrary colours.}
         \label{fig:7}

\end{figure}

We run our experiment with six different coloured graphs. For three of them, the skeleton is a cycle of order $p=20$ and for the other three, the skeleton is a cycle of order $p=30$. For each cycle of order $p$, we give three different patterns of colouring which, for the sake of saving space, are illustrated in Figure \ref{fig:7} for $p=6$. The values for the entries of $K$ for all three types of graphs are as follows:
 \begin{eqnarray*}
 &&K_{ii}=0.1, \;i=1,3,\ldots, 2p-1, \;\;K_{ii}=0.03.\;i=2,4,\ldots, p,\;\; \\
 &&K_{i,i+1}=K_{i+1,i}=0.01,\;i=1,2,\ldots, p-1,\;\;K_{1p}=K_{p1}=0.01.
 \end{eqnarray*}
 Though, for convenience, we chose, for all the models, the same values for the entries  $K_{ij}, i\not =j$ to be all equal to $.01$ and  for $K_{ii}$ to have two different values $.1$ and $.03$, in our computations, we used, of course, in each case, the model represented by each of the respective graphs.
 For each graph, we generated 100 datasets from the $N(0,K^{-1})$ distribution.
The posterior mean estimates are based on  5000 iterations after the first 1000 burn-in iterations. We denote $\hat{K}=(\hat{K}_{ij})_{1\leq i,j\leq p}$  the posterior mean estimate.

Table \ref{table:2} shows $\overline{nmse}(K,\hat{K})$ for the three colored models on the simulated examples when $p=20 $ and $p=30$, averaged over 100 simulations. Standard errors are indicated in parentheses.
Computations were performed on a 2 core 4 threads with i5-4200U, 2.3 GHZ chips and 8GB of RAM, running on Windows 8.  We also give in Table \ref{table:2} the average computing time per simulation in minutes.

\begin{table}[h!]
\centering
\caption{$\overline{nmse}(K,\hat{K})$ for the three colored models when $p=20$ and $p=30$.}
\begin{tabular}{ c| cc|cc}
     \hline
      &\multicolumn{2}{c}{$p=20$} &\multicolumn{2}{|c}{$p=30$}\\
     $\cal G$ & $ \overline{nmse}(K,\hat{K})$ & Time/sim   & $ \overline{nmse}(K,\hat{K})$ & Time/sim\\ \hline
      Fig.2 (a) &  0.005 (0.003) &19.425   &  0.040 (0.021)&86.423\\
      Fig.2 (b) &  0.011 (0.003) &18.739  &  0.033 (0.011)&82.876 \\
      Fig.2 (c) &  0.039 (0.021) &16.410   &  0.080 (0.033)&82.563\\
    \hline
   \end{tabular}
\label{table:2}
\end{table}
In Table \ref{table:3}, for the graph of Fig. \ref{fig:7} (a) with $p=20$ and $p=30$, we give the values of the entries of $K$ together with their batch standard errors.
\begin{table}[h!]
\centering
\caption{The average estimates and batch standard errors for $K$ in Fig.2 (a).}
\begin{tabular}{ c|cccc}
$p$&$K_{11}$&$K_{12}$&$K_{1p}$&$K_{22}$\\
\hline
      $20$ &  0.1040 (0.0005) &0.0103 (0.0002)   &  0.0104 (0.0002)&0.0313 (0.0001)\\
      $30$ &  0.1223 (0.0009) &0.0121 (0.0004)  &  0.0125 (0.0004)&0.0361 (0.0003) \\
   \end{tabular}
\label{table:3}
\end{table}
For the other models, average values of the entries together with batch standard errors are given in  Appendix 3 in the Supplementary file.
\vspace{3mm}

{\it Remark 1.} At this point, we ought to make an important remark. In Section 5, we proved that the $CG$-Wishart was proper for $D\in P_{\G}^*$. When we compute the posterior mean in this section or more generally for any colored graph, even if $D>0$ belongs to $P_{\G}^*$,  the hyperparameter $D+nS$ does not usually belong to $P^*_{\G}$ of course and yet the integral $I_{\G}(\delta+n, D+nS)$ converges. This is due to the fact that we can write $S$ as
$$S=S_1+S_2$$
where $S_1$ is the projection of $S$ on the subspace of $p\times p$ matrices with fixed zeros according to $\G$ and equal entries for edges and vertices in the same colour class and $S_2$ belongs to its orthogonal complement.
Since $K\in P_{\G}$, we have
$$0<\langle K,S\rangle=\langle K, S_1\rangle$$
and, since the inequality above is true for any $K\in P_{\G}$, it follows that $S_1$ belongs to $P_{\G}^*$. It follows also that
$I_{\G}(\delta+n,D+nS)$ is finite.
\vspace{2mm}

{\it Remark 2.} For the computation of the posterior mean following our sampling scheme of Section 3, we may wonder  whether we should take $Q$ to be such that $Q^tQ=(D+nS)^{-1}$ or $Q^tQ=(D+nS_1)^{-1}$. We take $Q$ to be such that $Q^tQ=(D+nS)^{-1}$ to use all the information given by the data.

\appendix

\section*{Appendix 1}
\subsection*{Proofs of Section 4}
\noindent {\bf Proof of Proposition \ref{pgtree}}
\newline Let $M$ be the set of $p\times p$ matrices.
Let
\footnotesize
$${\cal T}_{G}=\{X  \in M\mid X_{ij}=0,\;\mbox{for}\; i<j,\;X_{ij}=s_{ij}\neq 0, \;\mbox{for}\;i>j, (i,j)\in E,\; X_{ii}=t_i>0, i=1,\ldots,p\}$$
\normalsize
be the set of upper triangular matrices with positive diagonal elements and nonzero entries $X_{ij}, i>j$ only for $(i,j)\in E$.
The vector ${s}=(s_{ij}, (i,j)\in E)$ belongs to $R^{p-1}$ since a tree with $p$ vertices has $p-1$ edges and $t=(t_i,i=1,\ldots,p)$  belongs to $R^p$.
It is well-known (see \cite{Paul:1989} and \cite{Rov:2000}) that we can find a perfect elimination scheme enumeration of the vertices of $T$ such that, with this enumeration,
  $K\in P_{\cal G}$ can be written as $K=X(t,{s})^TX(t,{s})$ with $X(t,{s})\in {\cal T}_G$.
 Then for $K=K(a,b)$ as in \eqref{pgt}  we have
$$a_j=t_j^2+\sum_{i\in E_j}s_{ij}^2,\ b=t_is_{ij},$$
where $(t,s)$ is the Cholesky parametrization of $K\in P_{\cal G}$. We can also parametrize $K\in P_{\cal G}$ with $(t,b)\in (0, +\infty)^p\times R$
using
\begin{equation}\label{CRR}
a_j=t_j^2+b^2\sum_{i\in E_j}\frac{1}{t_i^2}.
\end{equation}
 In this proof and the following one, we assume  that the numbering of the vertices of $T$ follows a perfect elimination scheme ordering. We then say that the last vertex $p$ in that ordering is the root of the tree and we will write
 $$E_j=\{i,i<j\mid (i,j)\in E\}.$$
For convenience, we denote by $C$ the right-hand side of equation \eqref{pgtgiven}.

We show first that $P_{\cal G}^*\subset C.$ Let  $D=m(a',b')\in P_{\cal G}^*$. Using \eqref{CRR}, we have
\begin{equation}\label{CCR}\langle K, D\rangle=a_1a'_1+\cdots+a_pa'_p+2(p-1)bb'= t^2_p a'_p+Ab^2+2Bb+C>0,\end{equation} where
\begin{equation}\label{TR} A=\sum_{j=1}^p\left(\sum_{i\in E_j}\frac{1}{t_i^2}\right)a'_j,\  \ B=(p-1)b',\ \  C=\sum_{i=1}^{p-1}t_i^2a'_i.\end{equation}
Now observe that for fixed $i=1,\ldots,p$ then either $i=p$ and the set $\{j; i\in E_j\}$ is empty since $p$ is the root of the tree, or the set $\{j; i\in E_j\}$ is reduced to one point, say $j_i.$  Therefore we have $\sum_{i\in E_j}a'_j=a_{j_i}$ for $i<p$ and zero for $i=p.$ (For the graph in Figure \ref{fig:1} (a), we have  $j_1=j_2=j_5=j_6=7$ and $j_3=j_4=6$) and it follows that
\begin{equation}\label{AR}A=\sum_{i=1}^{p-1}\frac{1}{t_i^2}a'_{j_i}.\end{equation}

Let us prove that $a'_j>0$ for all $j=1,\ldots,p.$ Take  $(a_1,\ldots,a_p)\in [0,\infty)^p\setminus \{0,\ldots,0\}.$ Then
$K(a,0)\in \overline{P_{\cal G}}\setminus\{0\}$ and $\langle K, D\rangle=a_1a'_1+\cdots+a_pa'_p>0$ implies that  $a'_j>0$ for all $j.$
Let us  now prove that $Ab^2+2Bb+C\geq 0$ for all $b.$ If not, there exists $b_0$ such that $Ab_0^2+2Bb_0+C< 0.$ Since $a'_p>0$ taking $t_p$ very small and $b=b_0$ in (\ref{CCR}) gives a contradiction.

Let us prove that  \begin{equation}\label{NI}|b'|\leq \frac{1}{p-1}\sum_{(i,j)\in E} \sqrt{a'_ia'_j}.\end{equation} Since $\forall b,\;Ab^2+2Bb+C\geq 0$, we have $B^2\leq AC.$ Now consider the function
$$(t_1,\ldots,t_{p-1})\mapsto AC$$  and let us compute its minimum $A^*C^*$ on $(0,\infty)^{p}.$ This function $AC$ is homogeneous of degree 0 and therefore if its minimum is reached at $t^*=(t_1^*,\ldots,t_{p-1}^*)$ it  will also be reached on $\kappa t^*$ for any $\kappa>0.$
We have for $i=1,\ldots,p-1$  $$\frac{\partial}{\partial t_i}AC=2t_ia'_iA-\frac{2}{t_i^3}a'_{j_i}C=0$$ and we therefore have
$$t_i^*=\kappa\left(\frac{a'_{j_i}}{a'_i}\right)^{1/4}, \ A^*=C^*=\sum_{i=1}^{p-1}\sqrt{a'_ia'_{j_i}}=\sum_{(i,j)\in E} \sqrt{a'_ia'_j}.$$
Since $B^2\leq AC$ for all $(t_1,\ldots,t_{p-1})\in (0,\infty)^{p-1}$, we can claim that $B^2\leq A^*C^*$  or equivalently (\ref{NI}).

Let us prove that inequality (\ref{NI}) is strict, that is $B^2= A^*C^*$ is impossible. Suppose that $B^2= A^*C^*$, i.e.  $|b'|=A^*/(p-1)>0.$ Then with $t_i=t_i^*$ we get
$Ab^2+2Bb+C=A^*(b+ \sign b')^2.$
Taking $b=-\sign b'$ and $t_i=t_i^*, i=1,\ldots,p-1$ yields $Ab^2+2Bb+C=0.$ Now, letting  also  $t_p=0$ in (\ref{CCR}), we see that the left hand side of \eqref{CCR} is zero for an $(a,b)\in  \overline{S}\setminus\{0\}$ which is not zero, since $b=\pm 1.$ But this cannot happen for $D(a',b')\in P^*_{\cal G}.$
Therefore (\ref{NI}) is strict and  the proof of $P^*_{\cal G}\subset C$ is completed.

\vspace{4mm}\noindent Let us now show that $C\subset P_{\cal G}^*.$ For  $D(a',b')\in C$ given, we want to show that
$\langle K, D\rangle$ is positive for all  $K(a,b)\in \overline{ P_{\cal G}}\setminus \{0\}$. We will do so first for $K(a,b)\in  P_{\cal G}$ and then for $K(a,b)\in \overline{ P_{\cal G}}\setminus (P_{\cal G}\cup \{0\})$. For $K(a,b)\in  P_{\cal G}$,  $t_k>0$ and $b\in R.$
From  (\ref{CCR}), we have
$$\langle K, D\rangle=t_p^2a'_p+Ab^2+2Bb+C=t_p^2a'_p+ A\left[\left(b+\frac{B}{A}\right)^2+\frac{1}{A^2}(AC-B^2)\right]\;.$$ We have checked above that $AC-B^2\geq 0.$ Moreover $a'_p>0$ since $D(a',b')\in C$. It follows immediately that $\langle K, D\rangle>0.$

Let us now show $\langle K, D\rangle>0$ for $K(a,b)\in \overline{ P_{\cal G}}\setminus (P_{\cal G}\cup \{0\})$ that is
for $t_1\ldots t_p=0$ and $(t_1,\ldots,t_p,b)\neq 0.$ We need only show that then $\langle K, D\rangle\neq 0.$ But
$0=a_1a'_1+\cdots+a_pa'_p+2(p-1)bb'=\sum_{i=1}^pt_i^2a'_i$ implies that $t_i=0$ for all $i=1,\ldots,p$ since $(a',b')\in C$ implies $a'_i>0$. But since $b=t_is_{ij}$, this implies $b=0$ but this is impossible since we exclude the zero matrix for $K$.

\vspace{5mm}

\noindent {\bf Proof of Theorem \ref{pgtnorm:theorem}}

In $I_{\cal G}(\delta, D)$ we make the  change of variable \eqref{CRR}. Switching to these Cholesky coordinates  leads to the Jacobian  $dadb=2^{p}t_1\ldots t_p\, dbdt.$ As seen before the new domain of integration is the product $$\{(b,t); t_k>0, b\in R\}=(0,\infty)^{p}\times R.$$ With the notation $A,B,C$ of (\ref{TR}), we have
$$\langle K(a,b),D(a',b')\rangle =2(p-1)bb'+a_1a'_1+\cdots+a_pa'_p=t_p^2a'_p+Ab^2+2Bb+C.$$ Using (\ref{AR}) for the expression of $A$, we obtain

\begin{eqnarray}\nonumber I_{G}(\delta, D)&=&2^p\int_{(0,\, \infty)^{p}\times R}(t_1\ldots t_p)^{\delta-1}e^{-(p-1)bb'}e^{-\frac{t_p^2a'_p}{2}}\prod_{i=1}^{p-1}e^{-\frac{t_i^2a'_i}{2}-\frac{b^2a'_{j_i}}{2t_i^2}}dt_1\ldots dt_p\, db\\\nonumber
&=&2^p\int_{0}^{\infty}e^{-\frac{t_p^2a'_p}{2}}t_p^{\delta-1}dt_p \int_{-\infty}^{\infty}e^{-(p-1)bb'}\prod_{i=1}^{p-1}\left(K_{\delta/2}(|b|(a'_ia'_{j_i})^{1/2})(|b|\sqrt{a'_{j_i}/a_i})^{\delta/2}\right)db
\\&=&
2^{p+\frac{\delta}{2}-1}\frac{\Gamma(\delta/2)}{(a'_p)^{\delta/2}}\left(\prod_{i=1}^{p-1}\frac{a'_{j_i}}{a'_i}\right)^{\delta/4} J_{\delta}(D)\end{eqnarray}
with the notation
\begin{equation}\label{JDE} J_{\delta}(D)=\int_{-\infty}^{\infty}e^{-(p-1)bb'}|b|^{(p-1)\delta/2}\prod_{i=1}^{p-1}K_{\delta/2}(|b|(a'_ia'_{j_i})^{1/2})db.\end{equation}
We now prove by induction that
\begin{equation}\label{COT}\frac{1}{(a'_p)^2}\times \prod_{i=1}^{p-1}\frac{a'_{j_i}}{a_i}=\prod_{i=1}^{p}(a'_i)^{d_i-2}.\end{equation}
Of course (\ref{COT}) is correct for $p=2$. Suppose that (\ref{COT}) is true for any rooted tree with size $p.$ Consider a  rooted tree $T^*$  with vertices $\{0,1,\ldots,p\} $ and root $p$ and numbered, as usual, such that
$i\prec j$ implies $i\leq j$.  Denote $T$ the induced tree with vertices $\{1,\ldots,p\}$. Finally denote $d^*=(d_0^*,\ldots,d_p^*)$ and $d=(d_1,\ldots,d_p)$
the number of neighbours in $T^*$ and $T.$ Then $d_0^*=1$, $d_{j_0}^*=1+d_{j_0}$ and $d^*_i=d_i$ if $i\neq 0$ and $i\neq j_0.$
This implies that
$$\frac{1}{(a'_p)^2}\times \prod_{i=0}^{p-1}\frac{a'_{j_i}}{a_i}=\frac{a'_{j_0}}{a_0}\frac{1}{(a'_p)^2}\times \prod_{i=1}^{p-1}\frac{a'_{j_i}}{a_i}\stackrel{(1)}{=}\frac{a'_{j_0}}{a_0}\prod_{i=1}^{p}(a'_i)^{d_i-2}\stackrel{(2)}{=}\prod_{i=0}^{p}(a'_i)^{d^*_i-2},$$ where (1) comes from the induction hypothesis and (2) from the link between $d$ and $d^*.$ The induction hypothesis is extended to $p+1$ and (\ref{COT}) is proved.

We now prove that $J_{\delta}(D)$ defined by (\ref{JDE}) converges if $D=m(a',b')\in P^*_{\cal G}$ where $P^*_{\cal G}$ is the convex cone defined in Proposition 3.
We write $J_{\delta}(D)$ as the sum
\begin{equation}
\label{both}
J_{\delta}(D)=\int_{-\infty}^0\ldots db+\int_0^{+\infty}\ldots db.
\end{equation}

When $b\rightarrow \pm \infty$, $|b|\rightarrow +\infty$. From \cite{Wat:1995} page 202, 7.23 (1)
we have
$$K_{\lambda}(s)\sim_{s\rightarrow \infty} \sqrt{\frac{\pi}{2}}\frac{e^{-s}}{s^{1/2}}.$$
We use this fact to analyse the convergence of $J_{\delta}(D).$ If $D=m(a',b')\in P^*_{\cal G},$ from the  asymptotic formula above, we see that the integrands in both integral on the RHS of \eqref{both}, when $|b|$ goes to infinity, behave like $|b|^ce^{-|b|H}$ where, since $m(a',b')\in P^*_{\cal G}$,
 $$H=\sum_{(i,j)\in E}^p\sqrt{a'_ia'_j}-(p-1)|b'|\mathrm{sign}(bb')>0$$
  and $c=(p-1)\frac{\delta-1}{2}$. Since the argument of (\ref{JDE}) is continuous, both integrals converge at infinity.

 To study the convergence of these integrals when $b\rightarrow 0$, we recall that
$$2K_{\lambda}(s)=\int_0^{+\infty}x^{\lambda-1}e^{-\frac{s}{2}(x+\frac{1}{x})}dx.$$
 Making the change of variable $u=sx$ in the expression of  $2K_{\lambda}(s)$ we see that
  $$K_{\lambda}(s)\sim_{s\rightarrow 0} s^{-\lambda}2^{\lambda-1}\Gamma(\lambda).$$
  Therefore, for both integrals in the RHS of \eqref{both}, the integrand   is equivalent to
  $$\Big(|b|^{-\frac{\delta}{2}}\Gamma(\frac{\delta}{2})\Big)^{p-1}|b|^{\frac{p\delta}{2}}e^{-(p-1)bb'}=|b|^{\frac{\delta}{2}}e^{-(p-1)bb'}$$
  and therefore both integrals converge at $0$.
   The expression \eqref{pgtnorm} of the normalizing constant is now proved.

\vspace{3mm}

By \eqref{pgtnorm}, $I_{G}(1, D)=2^{p-\frac{1}{2}}\Gamma(\frac{1}{2})(a'_{p})^{-\frac{1}{2}}(\prod_{i=1}^{p-1}\frac{a'_{j_{i}}}{a_{i}'})^{\frac{1}{4}}J_{1}(D)$, where
\begin{eqnarray*}
J_{1}(D)&=&\int_{-\infty}^{\infty}e^{-(p-1)bb'}|b|^{\frac{p-1}{2}}\prod_{i=1}^{p-1} K_{\frac{1}{2}}(|b|(a'_{i}a'_{j_{i}})^{\frac{1}{2}})db\\
&=&\int_{-\infty}^{\infty}e^{-(p-1)bb'}|b|^{\frac{p-1}{2}}(\prod_{i=1}^{p-1}\sqrt{\frac{\pi}{2}}|b|^{-\frac{1}{2}}(a'_{i}a'_{j_{i}})^{-\frac{1}{4}}e^{|b|(a'_{i}a'_{j_{i}})^{\frac{1}{2}}})db\\
&=&(\frac{\pi}{2})^{\frac{p-1}{2}}\prod_{i=1}^{p-1}(a'_{i}a'_{j_{i}})^{-\frac{1}{4}}\int_{-\infty}^{\infty}e^{-(p-1)bb'-|b|\sum\limits^{p-1}_{i=1}(a'_{i}a'_{j_{i}})^{\frac{1}{2}}}db.
\end{eqnarray*}
We compute the integral
\begin{eqnarray*}
\int_{-\infty}^{\infty}e^{-(p-1)bb'-|b|\sum\limits^{p-1}_{i=1}(a'_{i}a'_{j_{i}})^{\frac{1}{2}}}db
&=&\int_{-\infty}^{0}e^{-(p-1)bb'+b\sum\limits^{p-1}_{i=1}(a'_{i}a'_{j_{i}})^{\frac{1}{2}}}db+\int_{0}^{\infty}e^{-(p-1)bb'-b\sum\limits^{p-1}_{i=1}(a'_{i}a'_{j_{i}})^{\frac{1}{2}}}db\\
&=&\frac{1}{\sum\limits^{p-1}_{i=1}(a'_{i}a'_{j_{i}})-(p-1)b'}+\frac{1}{\sum\limits^{p-1}_{i=1}(a'_{i}a'_{j_{i}})+(p-1)b'}\;.
\end{eqnarray*}
Therefore
\begin{eqnarray*}
J_{1}(D)&=&(\frac{\pi}{2})^{\frac{p-1}{2}}\prod_{i=1}^{p-1}(a'_{i}a'_{j_{i}})^{-\frac{1}{4}}\Big[\Big(\sum\limits^{p-1}_{i=1}(a'_{i}a'_{j_{i}})^{\frac{1}{2}}-(p-1)b'\Big)^{-1}+\Big(\sum\limits^{p-1}_{i=1}(a'_{i}a'_{j_{i}})^{\frac{1}{2}}+(p-1)b'\Big)^{-1}\Big].\\
\end{eqnarray*}
Since $\sum\limits^{p-1}_{i=1}(a'_{i}a'_{j_{i}})=\sum_{(i,j)\in E}(a'_{i}a'_{j})$, this yields the expression of $I_{\cal G}(1,D)$.
\vspace{3mm}

Similarly, from \eqref{pgtnorm}, $I_{G}(3,D)=2^{p+\frac{1}{2}}\Gamma(\frac{3}{2})(a'_{p})^{-\frac{3}{2}}\prod_{i=1}^{p-1}(\frac{a'_{j_{i}}}{a'_{i}})^{-\frac{3}{4}}J_{3}(D)$ with
\begin{eqnarray*}
J_{3}(D)&=&\int_{-\infty}^{\infty}e^{-(p-1)bb'}|b|^{\frac{3}{2}(p-1)}\prod_{i=1}^{p-1} K_{\frac{3}{2}}(|b|(a'_{i}a'_{j_{i}})^{\frac{1}{2}})db\\
&=&\int_{-\infty}^{\infty}e^{-(p-1)bb'}|b|^{\frac{3}{2}(p-1)}\prod_{i=1}^{p-1}[\sqrt{\frac{\pi}{2}}(|b|^{-\frac{1}{2}}(a'_{i}a'_{j_{i}})^{-\frac{1}{4}}+|b|^{-\frac{3}{2}}(a'_{i}a'_{j_{i}})^{-\frac{3}{4}})e^{-|b|(a'_{i}a'_{j_{i}})^{\frac{1}{2}}}]db\\
&=&\int_{-\infty}^{\infty}e^{-(p-1)bb'}|b|^{\frac{3}{2}(p-1)}(\frac{\pi}{2})^{\frac{p-1}{2}}|b|^{-\frac{3}{2}(p-1)}\prod_{i=1}^{p-1}(a'_{i}a'_{j_{i}})^{-\frac{3}{4}}\prod_{i=1}^{p-1}[(|b|(a'_{i}a'_{j_{i}})^{\frac{1}{2}}+1)e^{-|b|(a'_{i}a'_{j_{i}})^{\frac{1}{2}}}]db\\
&=&(\frac{\pi}{2})^{\frac{p-1}{2}}\prod_{i=1}^{p-1}(a'_{i}a'_{j_{i}})^{-\frac{3}{4}}\int_{-\infty}^{\infty}e^{-(p-1)bb'-|b|\sum\limits^{p-1}_{i=1}(a'_{i}a'_{j_{i}})^{\frac{1}{2}}}\prod_{i=1}^{p-1}(1+|b|(a'_{i}a'_{j_{i}})^{\frac{1}{2}})db\\
&=&(\frac{\pi}{2})^{\frac{p-1}{2}}\prod_{i=1}^{p-1}(a'_{i}a'_{j_{i}})^{-\frac{3}{4}}\int_{-\infty}^{\infty}e^{-(p-1)bb'-|b|\sum\limits^{p-1}_{i=1}(a'_{i}a'_{j_{i}})^{\frac{1}{2}}}(1+|b|\sigma_{1}+|b|^2\sigma_{2}+\ldots+|b|^{p-1}\sigma_{p-1})db\\
&=&(\frac{\pi}{2})^{\frac{p-1}{2}}\prod_{i=1}^{p-1}(a'_{i}a'_{j_{i}})^{-\frac{3}{4}}\sum\limits^{p-1}_{k=0}\sigma_{k}\int_{-\infty}^{\infty}e^{-(p-1)bb'-|b|\sum\limits^{p-1}_{i=1}(a'_{i}a'_{j_{i}})^{\frac{1}{2}}}|b|^kdb\;,
\end{eqnarray*}
where the $\sigma_i=\sigma_i(\sqrt{a_i'a'_{j_{i}}},i=1,\ldots,p-1)$ are the symmetric functions of $\sqrt{a_i'a'_{j_{i}}},i=1,\ldots,p-1$.
Since
\footnotesize
\begin{eqnarray*}
\int_{-\infty}^{\infty}e^{-(p-1)bb'-|b|\sum\limits^{p-1}_{i=1}(a'_{i}a'_{j_{i}})^{\frac{1}{2}}}|b|^{m-1}db
&=&\Gamma(m)\left[\left(\sum\limits^{p-1}_{i=1}(a'_{i}a'_{j_{i}})^{\frac{1}{2}}-(p-1)b'\right)^{-m}-\left(\sum\limits^{p-1}_{i=1}(a'_{i}a'_{j_{i}})^{\frac{1}{2}}+(p-1)b'\right)^{-m}\right],
\end{eqnarray*}
\normalsize
then
\footnotesize
\begin{eqnarray*}
I_{\G}(3,D)=2^{\frac{p}{2}-1}\pi^{\frac{p}{2}}\prod_{i=1}^{p}(a'_{i})^{-\frac{3}{2}}\sum\limits^{p-1}_{k=0}\sigma_{k}\Gamma(k+1)\left[\left(\sum\limits^{p-1}_{i=1}(a'_{i}a'_{j_{i}})^{\frac{1}{2}}-(p-1)b'\right)^{-k-1}
-\left(\sum\limits^{p-1}_{i=1}(a'_{i}a'_{j_{i}})^{\frac{1}{2}}+(p-1)b'\right)^{-k-1}\right].
\end{eqnarray*}
\normalsize
This yields the expression of $I_{\cal G}(3,D)$.
\vspace{5mm}

\noindent {\bf Proof of Proposition \ref{pgdaisyleaves}}
\newline By definition $P_{\cal G}^*=\{D=l(a',b',c')\in M({\cal G})\mid \langle K, D\rangle>0,\;\;K\in \bar{P}_{\cal G}\setminus \{0\}\}.$ Let $\beta$ denote the angle between $b$ and $b'$. Then, since $\cos \beta> -1$
\begin{eqnarray*}
 \langle K, D\rangle&=&aa'+ncc'+2||b||||b'||\cos \beta> aa'+ncc'-2||b||||b'||.
\end{eqnarray*}
Therefore $2||b||||b'||<aa'+ncc'$ and since $ac>0$, $\frac{4||b||^2||b'||^2}{ac}<\frac{(aa'+ncc')^2}{ac}$. By differentiation with respect to $a$ and $c$, we see that $\frac{(aa'+ncc')^2}{ac}\geq 4na'c'$ and therefore $\langle K, D\rangle>0$ implies that
$||b'||^2<na'c'.$
\vspace{5mm}

\noindent {\bf Proof of Theorem \ref{igdaisyleaves}.}
Let us introduce the matrix
$$A(r,s,t)=\left[\begin{array}{ccccc}r&s_1&s_2&\ldots&s_n\\0&t&0&\ldots&0\\0&0&t&\ldots&0\\\ldots&\ldots&\ldots&\ldots&\ldots\\0&0&0&\ldots&t\end{array}\right].$$ If $(a,b,c)\in P_{\cal G}$ the only triple $(r,s,t)$ such that $t>0$  and $r>0$ and such that
$$K(a,b,c)=A(r,s,t)A^{T}(r,s,t)=\left[\begin{array}{cc}r^2+\|s\|^2&s't \\ts&t^2I_n\end{array}\right]$$ satisfies
$r=(a-\frac{\|b\|^2}{c})^{1/2},\ t=\sqrt{c},\ s=\frac{b}{\sqrt{c}}.$
A new parameterization of $P_{\cal G}$ is therefore given by  the change of variables $(a,b,c)$ into $(r,s,t)$ with
$a=r^2+\|s\|^2,b=ts,\ c=t^2,$
where $(r,s,t)$ belongs to $$\{(r,s,t); r>0, s\in R^n,t>0\}=(0,\infty)\times R^{n}\times (0,\infty).$$ With this parameterization, from (\ref{DK}), we have $\det K=r^2t^{2n}$ and $dadbdc=4rt^{n+1}drdsdt.$
Then
\begin{eqnarray*}I_{\G}(\delta, D)&=&4\int_0^{\infty}\int_0^{\infty}\left(\int_{R^n}e^{\frac{-\|s\|^2a'-2t\<s,\, b'\>}{2}}ds\right)r^{\delta-1}t^{(\delta-1)n+1}e^{\frac{-r^2a'-nt^2c'}{2}}drdt\\&=&4\left(\frac{\pi}{a'}\right)^{n/2}\int_0^{\infty}e^{\frac{-nt^2c'}{2}+\frac{t^2\|b'\|^2}{2a'}}t^{(\delta-1)n+1}dt\times \int_0^{\infty}r^{\delta-1}e^{\frac{-r^2a'}{2}}dr\\&=&\left(\frac{\pi}{a'}\right)^{n/2}\int_0^{\infty}e^{\frac{-nvc'}{2}+\frac{v\|b'\|^2}{2a'}}v^{(\delta-1)\frac{n}{2}}dv\times \int_0^{\infty}v^{\frac{\delta}{2}-1}e^{\frac{-va'}{2}}dv\\&=&2^{\frac{\delta+n\delta+2}{2}}\pi^{n/2} a'^{(\frac{\delta}{2}-1)(n-1)} \frac{1}{(na'c'-\|b'\|^2)^{(\delta-1)\frac{n}{2}+1}}\Gamma((\delta-1)\frac{n}{2}+1)\Gamma(\frac{\delta}{2}).\end{eqnarray*}
\vspace{5mm}

\noindent {\bf Proof of Theorem \ref{igdaisyall}.}

\begin{eqnarray*}
I_{\cal G}(\delta,D)&=&\int_{R^n}\left(\int_{\|b\|}^{\infty}a^{n\frac{\delta-2}{2}}(a-\frac{||b||^2}{a})^{\frac{\delta-2}{2}}\exp -\frac{1}{2}\{(n+1)aa'+2 \langle b,b'\rangle\}da\right) db\\
&=&\int_{R^n}\left(\int_{\|b\|}^{\infty}a^{(n-1)\frac{\delta-2}{2}}(a^2-||b||^2)^{\frac{\delta-2}{2}}\exp -\frac{1}{2}\{(n+1)aa'+2 \langle b,b'\rangle\}da\right) db
\end{eqnarray*}
Let us make the change of variable
$$(a,b)\in (||b||,+\infty)\times R^n\mapsto (u, R, \theta)\in (0,1)\times (0, +\infty)\times S$$
 where $b=R\theta$ and $S$ is the unit sphere in $R^n$ and $a=\frac{R}{\sqrt{u}}$. We have $da db=-\frac{1}{2u^{3/2}}RC_nR^{n-1}da dR d \theta$ where $C_n$ is the surface area of $S$.
Then
\tiny
\begin{eqnarray*}
I_{\cal G}(\delta,D)&=&\frac{C_n}{2}\int_S\left[\int_{0}^{+\infty}\left(\int_0^1
R^{(n-1)\frac{\delta-2}{2}}u^{-(n-1)\frac{\delta-2}{4}}R^{\delta-2}(\frac{1}{u}-1)^{\frac{\delta-2}{2}}\exp -\{\frac{(n+1)Ra'}{2\sqrt{u}}+R\langle \theta,b'\rangle\}u^{-3/2}du\right)R^ndR\right]d\theta\\
&=&\frac{C_n}{2}\int_S\left[\int_{0}^{+\infty}\left(\int_0^1
R^{(n-1)\frac{\delta-2}{2}}u^{-(n+1)\frac{\delta-2}{4}}R^{\delta-2}(1-u)^{\frac{\delta-2}{2}}\exp -\{\frac{(n+1)Ra'}{2\sqrt{u}}+R\langle \theta,b'\rangle\}u^{-3/2}du\right)R^ndR\right]d\theta\\
&=&\frac{C_n}{2}\int_S\left[\int_{0}^{+\infty}\left(\int_0^1
R^{(n+1)\frac{\delta}{2}-1}u^{-(n+1)\frac{\delta-2}{4}-\frac{3}{2}}(1-u)^{\frac{\delta-2}{2}}\exp -R\{\frac{(n+1)a'}{2\sqrt{u}}+\langle \theta,b'\rangle\}du\right)dR\right]d\theta\\
&=&\frac{C_n}{2}\int_S\left[\int_0^1u^{-(n+1)\frac{\delta-2}{4}-\frac{3}{2}}(1-u)^{\frac{\delta-2}{2}}
\left(\int_{0}^{+\infty}R^{(n+1)\frac{\delta}{2}-1}\exp -R\{\frac{(n+1)a'}{2\sqrt{u}}+\langle \theta,b'\rangle\}dR\right)du\right]d \theta\\
&=&\frac{C_n\Gamma((n+1)\frac{\delta}{2})}{2}\int_S\left[\int_0^1u^{-(n+1)\frac{\delta-2}{4}-\frac{3}{2}}(1-u)^{\frac{\delta-2}{2}}\Big(\frac{(n+1)a'}{2\sqrt{u}}+\langle \theta,b'\rangle\Big)^{-(n+1)\frac{\delta}{2}}du\right]d \theta\\
&=&\frac{C_n\Gamma((n+1)\frac{\delta}{2})}{2}(\frac{(n+1)a'}{2})^{-(n+1)\frac{\delta}{2}}\int_S\left[\int_0^1u^{-(n+1)\frac{\delta-2}{4}-\frac{3}{2}}(1-u)^{\frac{\delta-2}{2}}u^{(n+1)\frac{\delta}{4}}\Big(1+\frac{2}{(n+1)a'}\sqrt{u}\langle \theta,b'\rangle\Big)^{-(n+1)\frac{\delta}{2}}du\right]d \theta\\
&=&K_{n,\delta}(a')\int_S\left[\int_0^1u^{\frac{n}{2}-1}(1-u)^{\frac{\delta}{2}-1}\sum_{k=0}^{\infty}(-1)^k\Big(\frac{2\langle \theta,b'\rangle}{(n+1)a'}\Big)^ku^{\frac{k}{2}}\frac{\Big((n+1)\frac{\delta}{2}\Big)_k}{k!}\;\;
du\right]d \theta
\end{eqnarray*}
\normalsize
where $K_{n,\delta}(a')=\frac{2^{\frac{(n+1)\delta}{2}-1}C_n\Gamma((n+1)\frac{\delta}{2})}{(n+1)^{\frac{(n+1)\delta}{2}}(a')^{\frac{(n+1)\delta}{2}}}$. Therefore
\begin{eqnarray*}
I_{\cal G}(\delta,D)&=&K_{n,\delta}(a')
\sum_{k=0}^{\infty}(-1)^k\Big(\frac{2}{(n+1)a'}\Big)^k\frac{\Big((n+1)\frac{\delta}{2}\Big)_k}{k!}
\int_0^1u^{\frac{k+n}{2}-1}(1-u)^{\frac{\delta}{2}-1}du\int_S\langle \theta,b'\rangle^k d\theta\\
&=&K_{n,\delta}(a')
\sum_{k=0}^{\infty}\Big(\frac{2}{(n+1)a'}\Big)^{2k}\frac{\Big((n+1)\frac{\delta}{2}\Big)_{2k}}{(2k)!}
\int_0^1u^{\frac{2k+n}{2}-1}(1-u)^{\frac{\delta}{2}-1}du\int_S\langle \theta,b'\rangle^{2k} d\theta\\
&=&K_{n,\delta}(a')
\sum_{k=0}^{\infty}\Big(\frac{2}{(n+1)a'}\Big)^{2k}\frac{\Big((n+1)\frac{\delta}{2}\Big)_{2k}}{(2k)!}
\frac{\Gamma(k+\frac{n}{2})\Gamma(\frac{\delta}{2})}{\Gamma(k+\frac{\delta+n}{2})}\|b'\|^{2k}\frac{(1/2)_k}{(n/2)_k}
\end{eqnarray*}
We now use the fact that $(\alpha)_{2k}=2^{2k}\Big(\frac{\alpha}{2}\Big)_{k}\Big(\frac{\alpha+1}{2}\Big)_{k}$ and $\Gamma(\alpha+k)=\Gamma(\alpha)(\alpha)_k$.
We also use the fact that
\footnotesize
$$(2k)!=(1 3 5...(2k-1))(2 4 6...2k)=2^kk!2^k\frac{1}{2}\frac{3}{2}...\frac{2k-1}{2}=2^{2k}k!\frac{1}{2}(\frac{1}{2}+1)(\frac{1}{2}+2)...(\frac{1}{2}+(k-1))=2^{2k}k!\Big(\frac{1}{2}\Big)_k.$$
\normalsize
Finally, since the integral is rotational symmetric, we take $b'=||b'||e_1$ so that  $\<\theta,b'\>=\theta_1||b'||$ and recalling that $d\theta$ is the distribution of $\frac{Z}{||Z||}$ when $Z\sim N(0,1)$ so that $\theta_1=\frac{Z_1}{\sqrt{Z_1^2+\ldots+Z_n^2}}$ which is then such that $\theta_1^2\sim \mbox{Beta}(\frac{1}{2}, \frac{n-1}{2})$, for $v=\theta_1$, we have
$$\int_S\<\theta,b'\>^{2k}d\theta=\frac{\|b'\|^{2k}}{B(\frac{1}{2},\frac{n-1}{2})}\int_0^1v^{k-\frac{1}{2}}(1-v)^{\frac{n-1}{2}-1}dv=\|b'\|^{2k}\frac{(1/2)_k}{(n/2)_k}.$$

Writing $B(\alpha, \beta)$ for the Beta function with argument $(\alpha,\beta)$, we obtain
\footnotesize
\begin{eqnarray*}
I_{\cal G}(\delta,D)&=&K_{n,\delta}(a')B(\frac{\delta}{2},\frac{n}{2})
\sum_{k=0}^{\infty}\Big(\frac{2}{(n+1)a'}\Big)^{2k}\frac{2^{2k}}{(2k)!}\Big((n+1)\frac{\delta}{4}\Big)_k\Big((n+1)\frac{\delta}{4}+\frac{1}{2}\Big)_k\frac{\Big(\frac{n}{2}\Big)_k}{\Big(\frac{n+\delta}{2}\Big)_k}\|b'\|^{2k}\frac{(1/2)_k}{(n/2)_k}\\
&=&K_{n,\delta}(a')B(\frac{\delta}{2},\frac{n}{2})
\sum_{k=0}^{\infty}\Big(\frac{2}{(n+1)a'}\Big)^{2k}\frac{2^{2k}}{2^{2k}k!(\frac{1}{2})_k}\Big((n+1)\frac{\delta}{4}\Big)_k\Big((n+1)\frac{\delta}{4}+\frac{1}{2}\Big)_k\frac{\Big(\frac{n}{2}\Big)_k}{\Big(\frac{n+\delta}{2}\Big)_k}\|b'\|^{2k}\frac{(1/2)_k}{(n/2)_k}\;.
\end{eqnarray*}
\normalsize
Let $u=\Big(\frac{2||b'||}{(n+1)a'}\Big)^2$.  We note that since $D=l(a',b',a')\in P_{\G}^*$, then $u\le 1$.
After obvious simplifications in the expression above, we have
\begin{eqnarray*}
I_{\cal G}(\delta,D)&=&K_{n,\delta}(a')B(\frac{\delta}{2},\frac{n}{2})\sum_{k=0}^{\infty}\frac{u^k}{k!}\frac{\Big((n+1)\frac{\delta}{4}\Big)_k\Big((n+1)\frac{\delta}{4}+\frac{1}{2}\Big)_k}{\Big(\frac{n+\delta}{2}\Big)_k}\\
&=&K_{n,\delta}(a')B(\frac{\delta}{2},\frac{n}{2}) \;_2F_1\Big((n+1)\frac{\delta}{4}, (n+1)\frac{\delta}{4}+\frac{1}{2},\frac{n+\delta}{2};u\Big)\;.
\end{eqnarray*}
\vspace{5mm}

\noindent {\bf Proof of Proposition \ref{pgtriangle}}
\newline We write the Cholesky decomposition of $K$ under the form $K=AA^t$ with
$$A=\left(\begin{array}{ccc}a_{11}&a_{12}&a_{13}\\0&a_{22}&a_{23}\\0&0&a_{33}\end{array}\right).$$
Expressing the $k_{ij}$ in terms of the $a_{ij}$ and imposing $k_{13}=k_{23}$ immediately shows that we must have $a_{13}=a_{23}$. Then, let $D=(d_{ij})_{1\leq i,j\leq 3}$ with $d_{13}=d_{23}$ since the dual of $P_{\cal G}$ must be in the same linear space as $P_{\cal G}$.
\begin{eqnarray*}
\langle K,D\rangle&=&(a_{11}^1+a_{12}^2+a_{13}^2)d_{11}+(a_{22}^2+a_{13}^2)d_{22}+a_{33}^2d_{33}+2(a_{22}a_{12}+a_{13}^2)d_{12}+4a_{13}a_{33}d_{13}\\
&=&a_{13}^2(d_{11}+d_{22}+2d_{12})+4a_{13}a_{33}d_{13}+a_{12}^2d_{11}+2a_{12}a_{22}d_{12}+a_{11}^2d_{11}+a_{22}^2d_{22}+a_{33}^2d_{33}\;,
\end{eqnarray*}
which we view as a quadratic form $a^tMa$ with $a^t=(a_{13},a_{33},a_{12},a_{22},a_{33})$ and
$$M=\left(\begin{array}{ccrrr}d_{11}+d_{22}+2d_{12}&2d_{13}&0&0&0\\2d_{13}&d_{33}&0&0&0\\
0&0&d_{11}&d_{12}&0\\0&0&d_{12}&d_{22}&0\\0&0&0&0&d_{11}\end{array}\right).$$
Since $AA^t$ is the Cholesky parametrization of $P_{\cal G}$, clearly $K\in P_{\cal G}$ if and only if $a_{ii}>0,i=1,2,3$. If we can prove the following lemma, the condition $M>0$ will yield the dual cone $P^*_{\cal G}$.

\text{\emph{Lemma A1.}} The trace $\langle K,D\rangle$ is positive for all $K\in \bar{P}_{\cal G}\setminus \{0\}$ if and only if the matrix $M$ of the quadratic form $\langle K, D\rangle=a^tMa$ is positive definite

Let us now prove the lemma. Clearly if $M>0$ then $\langle K,D\rangle=a^tMa>0$ for all $a\in R^5$ and in particular for all $a$ with $a_{ii}>0, i=1,2,3$. Conversely let $a\in R^5$. Then $a$ can be written as
$$a=(\epsilon_1a_{11}, \epsilon_2a_{22}, \epsilon_3a_{33},a_{12},a_{13})^t$$
where $\epsilon_i$ is the sign of $a_{ii},i=1,2,3$ and we have
$$a^tMa=(a_{11}^1+a_{12}^2+a_{13}^2)d_{11}+(a_{22}^2+a_{13}^2)d_{22}+a_{33}^2d_{33}+2(\epsilon_2a_{22}a_{12}+a_{13}^2)d_{12}+4\epsilon_3a_{13}a_{33}d_{13}.$$
But this is also equal to $\tilde{a}^tM\tilde{a}$ where
$$\tilde{a}^t=(|a_{11}|, |a_{22|}, |a_{33}|,\epsilon_2a_{12},\epsilon_3a_{13})$$
which is in $P_{\cal G}$. Therefore $\langle K,D\rangle>0$ for all $K\in P_{\cal G}$ if and only if $M$ is positive definite which translates immediately into the conditions defining $P^*_{\cal G}$ in Proposition \ref{pgtriangle}.
\vspace{5mm}

\noindent {\bf Proof of Theorem \ref{igtriangle}}
For the proof of the theorem, it will be convenient to adopt a slightly different form of the parametrization of the Cholesky decomposition of $K=AA^t$ in $P_{\cal G}$.
Let
\[ A_{ij} = \left\{
   \begin{array}{l l}
     \sqrt{a_{ii}} & \quad \text{if $i=j$,}\\
     -a_{ij} & \quad \text{if $i < j$.}
   \end{array} \right.\]
so that
\[ (AA^{T})_{ij} = \left\{
   \begin{array}{l l}
     a_{ii}+\sum\limits_{l>i}a^2_{il} & \quad \text{if $i=j$,}\\
     -a_{ij}\sqrt{a_{jj}}+\sum\limits_{l>max(i,j)}a_{il}a_{jl} & \quad \text{if $i < j$.}
   \end{array} \right.\]

 Equating each entry $k_{ij}$ of $K$ to the corresponding entry of $AA^{T}$ with the constraint that $k_{13}=k_{23}$ shows that

$k_{11}=a_{11}+a^2_{12}+a^2_{13}$, \ \ \ \ $k_{12}=-\sqrt{a_{22}}a_{12}+a_{13}a_{23}$,

$k_{22}=a_{22}+a^2_{23}$, \ \ \ \ \ \ \ \ \ \ \ \ $k_{13}=-\sqrt{a_{33}}a_{13}$,

$k_{33}=a_{33}$,\ \ \ \ \ \ \ \ \ \ \ \ \ \ \ \ \ \ \ \ \ $k_{23}=-\sqrt{a_{33}}a_{23}$.

In particular, we find that since $a_{33}>0$, $a_{13}=a_{23}$ and $k_{12}=-\sqrt{a_{22}}a_{12}+a^2_{13}$.
The Jacobian of the transformation from $K$ to $A$ is

$$J=\bordermatrix{~& k_{11} & k_{12} & k_{13} & k_{22} &k_{33} \cr
                  a_{11} & 1 & 0& 0&0&0\cr
                  a_{12} & *&-\sqrt{a_{22}} & 0&0 &0\cr
                  a_{13} & *&* &-\sqrt{a_{33}} & 2a_{13}&0\cr
                  a_{22} & *&*&* &1 &0\cr
                  a_{33} & *& *&*& *&1\cr}
$$

It is easy to see $|J|=|diag(J)| = a^{1/2}_{22}a^{1/2}_{33}$.

We now have all the ingredients necessary to calculate the normalizing constant $I_{\cal G}(\delta, D)$.
We have
$|K|=a_{11}a_{22}a_{33}$ and
\begin{eqnarray*}
\langle K,D \rangle&=&d_{11}k_{11}+d_{22}k_{22}+d_{33}k_{33}+2d_{12}k_{12}+2d_{13}k_{13}+2d_{23}k_{23}\\
&=&d_{11}(a_{11}+a^2_{12}+a^2_{13})+d_{22}(a_{22}+a^2_{23})+d_{33}a_{33}\\
&&\hspace{.5cm}+ 2d_{12}(-\sqrt{a_{22}}a_{12}+a_{13}a_{23})+2d_{13}(-a_{13}\sqrt{a_{33}})+2d_{23}(-a_{23}\sqrt{a_{33}}).
\end{eqnarray*}
and so the normalizing constant is
\begin{eqnarray*}
I_{\cal G}(\delta, D)&=&\int_A a^{\frac{\delta-2}{2}}_{11}a^{\frac{\delta-1}{2}}_{22}a^{\frac{\delta-1}{2}}_{33}\exp(-\frac{1}{2}d_{11}a_{11}-\frac{1}{2}d_{22}a_{22}-\frac{1}{2}d_{33}a_{33}
-\frac{1}{2}d_{11}a^2_{12}\\
&&\hspace{.5cm}-\frac{1}{2}(d_{11}+d_{22}+2d_{12})a^2_{13}+d_{12}\sqrt{a_{22}}a_{12}+2d_{13}a_{13}\sqrt{a_{33}})dA.
\end{eqnarray*}
where $a_{ii}>0$; $a_{ij}\in R,i<j$; and $dA$ denotes the product of all differentials.
The integral with respect to $a_{11}$ is a gamma integral  with
\begin{eqnarray*}
\int^{\infty}_0 a^{\frac{\delta-2}{2}}_{11}\exp(-\frac{1}{2}d_{11}a_{11})da_{11}=2^{\frac{\delta}{2}}\Gamma(\frac{\delta}{2})d^{-\frac{\delta}{2}}_{11}.
\end{eqnarray*}
The integral with respect to $a_{12}$ and $a_{13}$ are Gaussian integrals with
\begin{eqnarray*}
\int^{\infty}_{-\infty} \exp(-\frac{1}{2}d_{11}a^2_{12}+d_{12}\sqrt{a_{22}}a_{12})da_{12}=\frac{\sqrt{2\pi}}{\sqrt{d_{11}}}\exp(\frac{d^2_{12}a_{22}}{2d_{11}}),
\end{eqnarray*}
and
\footnotesize
\begin{eqnarray*}
\int^{\infty}_{-\infty} \exp(-\frac{1}{2}(d_{11}+d_{22}+2d_{12})a^2_{13}+2d_{13}\sqrt{a_{33}}a_{13})da_{13}=\frac{\sqrt{2\pi}}{\sqrt{d_{11}+d_{22}+2d_{12}}}\exp(\frac{2d^2_{13}a_{33}}{d_{11}+d_{22}+2d_{12}}).
\end{eqnarray*}
\normalsize
Therefore
\footnotesize
\begin{eqnarray*}
I_{\cal G}(\delta, D)&=&\Gamma(\frac{\delta}{2})2^{\frac{\delta}{2}}d^{-\frac{\delta+1}{2}}_{11}2\pi(d_{11}+d_{22}+2d_{12})^{-\frac{1}{2}}\\
&&\hspace{.5cm}\int^{\infty}_{0}a^{\frac{\delta-1}{2}}_{22}a^{\frac{\delta-1}{2}}_{33}\exp\{(-\frac{1}{2}d_{22}+\frac{d^2_{12}}{2d_{11}})a_{22}+
(-\frac{1}{2}d_{33}+\frac{2d^2_{13}}{d_{11}+d_{22}+2d_{12}})a_{33}\}da_{22}da_{33}\\
&=&\Gamma(\frac{\delta}{2})2^{\frac{\delta}{2}}d^{-\frac{\delta+1}{2}}_{11}2\pi(d_{11}+d_{22}+2d_{12})^{-\frac{1}{2}}\\
&&\hspace{.5cm}\times \Gamma(\frac{\delta+1}{2})(\frac{2d_{11}}{d_{11}d_{22}-d^2_{12}})^{\frac{\delta+1}{2}}
\Gamma(\frac{\delta+1}{2})(\frac{2(d_{11}+d_{22}+2d_{12})}{d_{33}(d_{11}+d_{22}+2d_{12})-4d^2_{13}})^{\frac{\delta+1}{2}}\\
&=&\Gamma(\frac{\delta}{2})\Gamma^2(\frac{\delta+1}{2})\pi 2^{\frac{3\delta+4}{2}}(d_{11}+d_{22}+2d_{12})^{\frac{\delta}{2}}[d_{33}(d_{11}+d_{22}+2d_{12})-4d^2_{13}]^{-\frac{\delta+1}{2}}\\
&&\hspace{.5cm}\times (d_{11}d_{22}-d^2_{12})^{-\frac{\delta+1}{2}}\;.
\end{eqnarray*}
\normalsize
\vspace{5mm}

\noindent {\bf Proof of Proposition \ref{pgdecomp}}
\newline We proceed as in the proof of Proposition \ref{pgtriangle}. That is, we let $K=AA^t$ be the Cholesky decomposition of $K$ with $A$ upper triangular. Equating the entries of $K$ and $AA^t$ yields
$$a_{23}=a_{13},\;a_{24}=a_{14},\;a_{44}=a_{33}$$
with then
\begin{align*}
&k_{11}=a_{11}^2+a_{12}^2+a_{12}^2+a_{14}^2,\;&k_{12}=a_{12}a_{22}+a_{13}^2+a_{14}^2,\;&k_{13}=a_{13}a_{33}.\;&k_{14}=a_{14}a_{33}\\
& &k_{22}=a_{22}^2+a_{13}^2+a_{14}^2,\;& k_{23}=a_{13}a_{33},&k_{14}=a_{14}a_{33}\\
&&&k_{33}=a_{33}^2&k_{34}=0\\
&&&&k_{44}=a_{33}^2
\end{align*}
Then, ordering $\langle K, D\rangle$ as a polynomial in $a_{ij}$, we see that
\begin{eqnarray*}
\langle K, D\rangle&=&d_{11}a_{11}^2+d_{22}a_{22}^2+2d_{33}a_{33}^2+d_{11}a_{12}^2+2d_{12}a_{22}a_{12}+
a_{13}^2(d_{11}+2d_{12}+d_{22})\\
&&\hspace{2cm}+4d_{13}a_{13}a_{33}+a_{14}^2(d_{11}+2d_{12}+d_{22})+4d_{14}a_{14}a_{33}
\end{eqnarray*}
is a quadratic form and the matrix of this quadratic form is
$$M=\left(\begin{array}{cccccc}d_{11}&0&0&0&0&0\\0&d_{22}&d_{12}&0&0&0\\0&d_{12}&d_{11}&0&0&0\\
0&0&0&2d_{33}&2d_{13}&2d_{14}\\0&0&0&2d_{13}&d_{11}+2d_{12}+d_{22}&0\\0&0&0&2d_{14}&0&d_{11}+2d_{12}+d_{22}\end{array}\right).$$
With exactly the same argument as in Proposition \ref{pgtriangle}, we can show that $\langle K, D\rangle>0$ for all $K\in \bar{P}_{\cal G}$ if and only if $M>0$, i.e. $D$ satisfies the conditions of Proposition \ref{pgdecomp}.
\vspace{5mm}

\noindent {\bf Proof of Theorem \ref{igdecomp}}
As in the proof of Theorem \ref{igtriangle}, it will be convenient to adopt a slightly different parametrization of the Cholesky decomposition of $K$. Let
\[ A_{ij} = \left\{
   \begin{array}{l l}
     \sqrt{a_{ii}} & \quad \text{if $i=j$,}\\
     -a_{ij} & \quad \text{if $i < j$.}
   \end{array} \right.\]
so that the entries of $AA^t$ are given by
\[ (AA^{T})_{ij} = \left\{
   \begin{array}{l l}
     a_{ii}+\sum\limits_{l>i}a^2_{il} & \quad \text{if $i=j$,}\\
     -a_{ij}\sqrt{a_{jj}}+\sum\limits_{l>\mbox{max}(i,j)}a_{il}a_{jl} & \quad \text{if $i < j$.}
   \end{array} \right.\]
Equating  each entry $k_{ij}$ of $K$ to the corresponding entry of $AA^{T}$, we find that

$k_{11}=a_{11}+a^2_{12}+a^2_{13}+a^2_{14}$, \ \ \ \ $k_{12}=-\sqrt{a_{22}}a_{12}+a_{13}a_{23}+a_{14}a_{24}$,

$k_{13}=-\sqrt{a_{33}}a_{13}+a_{14}a_{34}$,   \ \ \ \ \ \ $k_{14}=-\sqrt{a_{44}}a_{14}$,

$k_{22}=a_{22}+a^2_{23}+a^2_{24}$, \ \ \ \  \ \ \ \ \ \ \ $k_{23}=-\sqrt{a_{33}}a_{23}+a_{24}a_{34}$,

$k_{24}=-\sqrt{a_{44}}a_{24}$ \ \ \ \ \ \ \ \ \ \ \ \ \  \ \ \ \  \ $k_{33}=a_{33}+a^2_{34}$,

$k_{34}=-\sqrt{a_{44}}a_{34}$, \ \ \ \ \ \ \ \ \ \ \ \ \ \ \ \ \ $k_{44}=a_{44}$.

This shows that $a_{44}>0$ and  $a_{34}=0$. Since $a_{33}>0$ and $k_{13}=k_{23}$, then $a_{13}=a_{23}$. Since $a_{44}>0$ and $k_{14}=k_{24}$, then $a_{14}=a_{24}$. Since $k_{34}=0$, then $a_{33}=a_{44}$. Therefore, we obtain that

$k_{11}=a_{11}+a^2_{12}+a^2_{13}+a^2_{14}$, \ \ \ \ $k_{12}=-\sqrt{a_{22}}a_{12}+a^2_{13}+a^2_{14}$,

$k_{13}=k_{23}=-\sqrt{a_{33}}a_{13}$, \ \ \ \ \ \ \ \ \ \ \ $k_{14}=k_{24}=-\sqrt{a_{33}}a_{14}$,

$k_{22}=a_{22}+a^2_{13}+a^2_{14}$, \ \ \ \ \ \ \ \ \ \ \ \ $k_{33}=k_{44}=a_{33}$.

The Jacobian of the transformation from $K$ to $A$ is

$$J=\bordermatrix{~& k_{11} & k_{12} & k_{13} & k_{14} & k_{22} &k_{33} \cr
                  a_{11} & 1 & 0& 0& 0&0&0\cr
                  a_{12} & *&-\sqrt{a_{22}} &0& 0&0 &0\cr
                  a_{13} & *&* &-\sqrt{a_{33}} & 0& 2a_{13}&0\cr
                  a_{14} & *&*&* &-\sqrt{a_{33}}&2a_{14} &0\cr
                  a_{22} & *& *&*& *&1&0\cr
                  a_{33} & *&* &*&* &*&1\cr}
$$

It is easy to see $|J|=|diag(J)|= a^{1/2}_{22}a_{33}$.
We now have all the ingredients necessary to calculate the normalizing constant $I_{G_{2}}(\delta, D)$.
Through the change of variables, $K=AA^{T}$. Then
$|K|=a_{11}a_{22}a^2_{33}$,
\footnotesize
\begin{eqnarray*}
\langle K,D \rangle&=&d_{11}k_{11}+d_{22}k_{22}+d_{33}k_{33}+d_{44}k_{44}+2d_{12}k_{12}+2d_{13}k_{13}+2d_{14}k_{14}+2d_{23}k_{23}+2d_{24}k_{24}+2d_{34}k_{34}\\
&=&d_{11}(a_{11}+a^2_{12}+a^2_{13}+a^2_{14})+d_{22}(a_{22}+a^2_{13}+a^2_{14})+2d_{33}a_{33}\\
&&\hspace{.5cm}+ 2d_{12}(-\sqrt{a_{22}}a_{12}+a^2_{13}+a^2_{14})+4d_{13}(-a_{13}\sqrt{a_{33}})+4d_{14}(-a_{14}\sqrt{a_{33}}).
\end{eqnarray*}
\normalsize
and so the integral equals
\footnotesize
\begin{eqnarray*}
I_{\cal G}(\delta, D)&=&\int_A a^{\frac{\delta-2}{2}}_{11}a^{\frac{\delta-1}{2}}_{22}a^{\delta-1}_{33}\exp\{-\frac{1}{2}d_{11}a_{11}-\frac{1}{2}d_{11}a^2_{12}
+d_{12}\sqrt{a_{22}}a_{12}-\frac{1}{2}(d_{11}+d_{22}+2d_{12})a^2_{13}\\
&&\hspace{.5cm}+2d_{13}a_{13}\sqrt{a_{33}}-\frac{1}{2}(d_{11}+d_{22}+2d_{12})a^2_{14}+2d_{14}a_{14}\sqrt{a_{33}}-\frac{1}{2}d_{22}a_{22}-d_{33}a_{33}\}dA.
\end{eqnarray*}
\normalsize
where $a_{ii}>0$; $a_{ij}\in R,i<j$; and $dA$ denotes the product of all differentials.
The integral with respect to $a_{11}$ is gamma integrals, then
\footnotesize
\begin{eqnarray*}
\int^{\infty}_0 a^{\frac{\delta-2}{2}}_{11}\exp(-\frac{1}{2}d_{11}a_{11})da_{11}=2^{\frac{\delta}{2}}\Gamma(\frac{\delta}{2})d^{-\frac{\delta}{2}}_{11}.
\end{eqnarray*}
\normalsize
The integral with respect to $a_{12}$, $a_{13}$ and $a_{14}$ are normal integrals, then

\begin{eqnarray*}
\int^{\infty}_{-\infty} \exp(-\frac{1}{2}d_{11}a^2_{12}+d_{12}\sqrt{a_{22}}a_{12})da_{12}=\frac{\sqrt{2\pi}}{\sqrt{d_{11}}}\exp(\frac{d^2_{12}a_{22}}{2d_{11}}),
\end{eqnarray*}

\footnotesize
\begin{eqnarray*}
\int^{\infty}_{-\infty} \exp(-\frac{1}{2}(d_{11}+d_{22}+2d_{12})a^2_{13}+2d_{13}\sqrt{a_{33}}a_{13})da_{13}=\frac{\sqrt{2\pi}}{\sqrt{d_{11}+d_{22}+2d_{12}}}\exp(\frac{2d^2_{13}a_{33}}{d_{11}+d_{22}+2d_{12}}),
\end{eqnarray*}
\normalsize
and
\footnotesize
\begin{eqnarray*}
\int^{\infty}_{-\infty} \exp(-\frac{1}{2}(d_{11}+d_{22}+2d_{12})a^2_{14}+2d_{14}\sqrt{a_{33}}a_{14})da_{14}=\frac{\sqrt{2\pi}}{\sqrt{d_{11}+d_{22}+2d_{12}}}\exp(\frac{2d^2_{14}a_{33}}{d_{11}+d_{22}+2d_{12}}).
\end{eqnarray*}
\normalsize
Therefore, the integral becomes
\footnotesize
\begin{eqnarray*}
I_{G_{1}}(\delta, D)&=&\Gamma(\frac{\delta}{2})2^{\frac{\delta}{2}}d^{-\frac{\delta+1}{2}}_{11}(2\pi)^{\frac{3}{2}}(d_{11}+d_{22}+2d_{12})^{-1}\\
&&\hspace{.5cm}\int^{\infty}_{0}a^{\frac{\delta-1}{2}}_{22}a^{\frac{\delta-1}{2}}_{33}\exp\{(-\frac{1}{2}d_{22}+\frac{d^2_{12}}{2d_{11}})a_{22}+
(-d_{33}+\frac{2d^2_{13}+2d^2_{14}}{d_{11}+d_{22}+2d_{12}})a_{33}\}da_{22}da_{33}\\
&=&\Gamma(\frac{\delta}{2})2^{\frac{\delta+3}{2}}d^{-\frac{\delta+1}{2}}_{11}\pi^{\frac{3}{2}}(d_{11}+d_{22}+2d_{12})^{-1}\\
&&\hspace{.5cm}\Gamma(\frac{\delta+1}{2})(\frac{2d_{11}}{d_{11}d_{22}-d^2_{12}})^{\frac{\delta+1}{2}}
\Gamma(\delta)(\frac{d_{11}+d_{22}+2d_{12}}{d_{33}(d_{11}+d_{22}+2d_{12})-2(d^2_{13}+d^2_{14})})^{\delta}\\
&=&\Gamma(\frac{\delta}{2})\Gamma(\frac{\delta+1}{2})\Gamma(\delta)\pi^{\frac{3}{2}} 2^{\delta+2}(d_{11}+d_{22}+2d_{12})^{\delta-1}[d_{33}(d_{11}+d_{22}+2d_{12})-2(d^2_{13}+d^2_{14})]^{-\delta}\\
&&\hspace{.5cm}(d_{11}d_{22}-d^2_{12})^{-\frac{\delta+1}{2}}\;.
\end{eqnarray*}
\normalsize

\section*{Appendix 2}
\subsection*{Numerical values for $D, E(K)$ and $\hat{K}$ and plots for Section 5}
We give here the matrices $D$, $E(K)$ and $\hat{K}$ as well as the traceplot and autocorrelation plot of $\log (|K|)$ for all graphs  in Fig. \ref{fig:1}. Here $\hat{K}$  has been computed with 5000 iterations after a 1000 iterations burn in and averaged over 100 simulations


\vspace{3mm}

\noindent {\bf Graph in Fig. \ref{fig:1}(a)}
\vspace{3mm}

$D=\left(\begin{array}{ccccccc}1&0&0&2&0&0&0\\0&2&0&2&0&0&0\\0&0&5&2&0&0&0\\2&2&2&25&2&0&0\\0&0&0&2&6&2&2\\
0&0&0&0&2&3&0\\0&0&0&0&2&0&4\end{array}\right)$,\;\;

$E(K)=\left(\begin{array}{ccccccc}1.1294&0&0&-0.0129&0&0&0\\0&0.5915&0&-0.0129&0&0&0\\0&0&0.2578&-0.0129&0&0&0\\-0.0129&-0.0129&-0.0129&0.0767&-0.0129&0&0\\
0&0&0&-0.0129&0.2589&-0.0129&-0.0129\\0&0&0&0&-0.0129&0.3699&0\\0&0&0&0&-0.0129&0&0.2817\end{array}\right),$

$\hat{K}=\left(\begin{array}{ccccccc}1.1274&0&0&-0.0127&0&0&0\\0&0.5961&0&-0.0127&0&0&0\\0&0&0.2563&-0.0127&0&0&0\\-0.0127&-0.0127&-0.0127&0.0767&-0.0127&0&0\\
0&0&0&-0.0127&0.2594&-0.0127&-0.0127\\0&0&0&0&-0.0127&0.3708&0\\0&0&0&0&-0.0127&0&0.2818\end{array}\right).$

\begin{figure}
         \centering
         \begin{subfigure}[b]{0.3\textwidth}
                 \includegraphics[width=\textwidth]{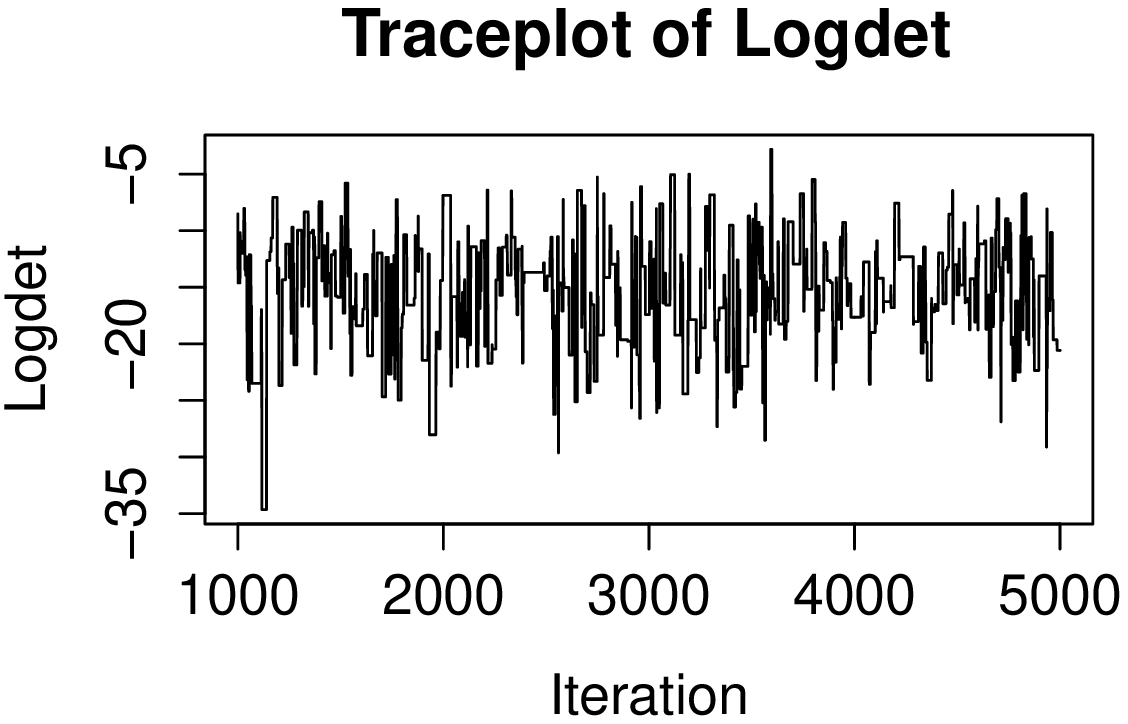}
                 \caption{Traceplot}
                 \label{fig:(a)trace}
         \end{subfigure}%
         ~ 
         \begin{subfigure}[b]{0.3\textwidth}
                 \includegraphics[width=\textwidth]{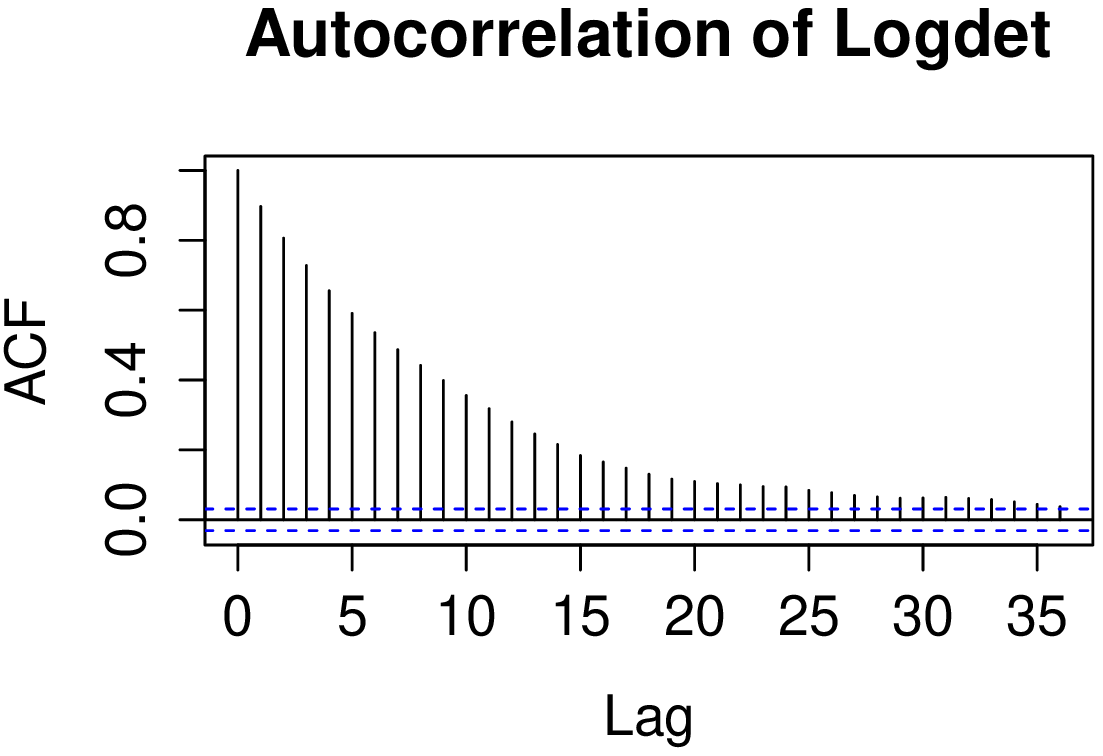}
                 \caption{ACF plot}
                 \label{fig:(a)auto}
         \end{subfigure}
~
         \caption{(a) Traceplot of $\log (|K|)$ v.s. the number of iterations. (b) Autocorrelation plot of $\log (|K|)$ for Fig. \ref{fig:1}(a). }\label{fig:TA1}
\end{figure}

\newpage

\noindent {\bf Graph in Fig. \ref{fig:1}(b)}
\vspace{3mm}

$D=\left(\begin{array}{ccccccccc}9&1&2&3&4&5&6&7&8\\1&25&0&0&0&0&0&0&0\\2&0&25&0&0&0&0&0&0\\
3&0&0&25&0&0&0&0&0\\4&0&0&0&25&0&0&0&0\\
5&0&0&0&0&25&0&0&0\\6&0&0&0&0&0&25&0&0\\7&0&0&0&0&0&0&25&0\\8&0&0&0&0&0&0&0&25\end{array}\right).$
\vspace{2mm}

\scriptsize

$E(K)=\left(\begin{array}{ccccccccc}1.4778&-0.0112&-0.0225&- 0.0338&-0.0451&-0.0563&-0.0676&-0.0789&-0.0902\\-0.0112&0.1015&0&0&0&0&0&0&0\\-0.0225&0&0.1015&0&0&0&0&0&0\\
- 0.0338&0&0&0.1015&0&0&0&0&0\\-0.0451&0&0&0&0.1015&0&0&0&0\\
-0.0563&0&0&0&0&0.1015&0&0&0\\-0.0676&0&0&0&0&0&0.1015&0&0\\-0.0789&0&0&0&0&0&0&0.1015&0\\-0.0902&0&0&0&0&0&0&0&0.1015\end{array}\right).$
\normalsize
\vspace{2mm}

\scriptsize

$\hat{K}=\left(\begin{array}{ccccccccc}1.4690&-0.0113&-0.0223&-0.0341&-0.0455&-0.0569&-0.0677&-0.0796&-0.0905\\-0.0113&0.1016&0&0&0&0&0&0&0\\-0.0223&0&0.1016&0&0&0&0&0&0\\
-0.0341&0&0&0.1016&0&0&0&0&0\\-0.0455&0&0&0&0.1016&0&0&0&0\\
-0.0569&0&0&0&0&0.1016&0&0&0\\-0.0677&0&0&0&0&0&0.1016&0&0\\-0.0796&0&0&0&0&0&0&-0.1016&0\\-0.0905&0&0&0&0&0&0&0&-0.1016\end{array}\right).$
\normalsize

\begin{figure}
         \centering
         \begin{subfigure}[b]{0.3\textwidth}
                 \includegraphics[width=\textwidth]{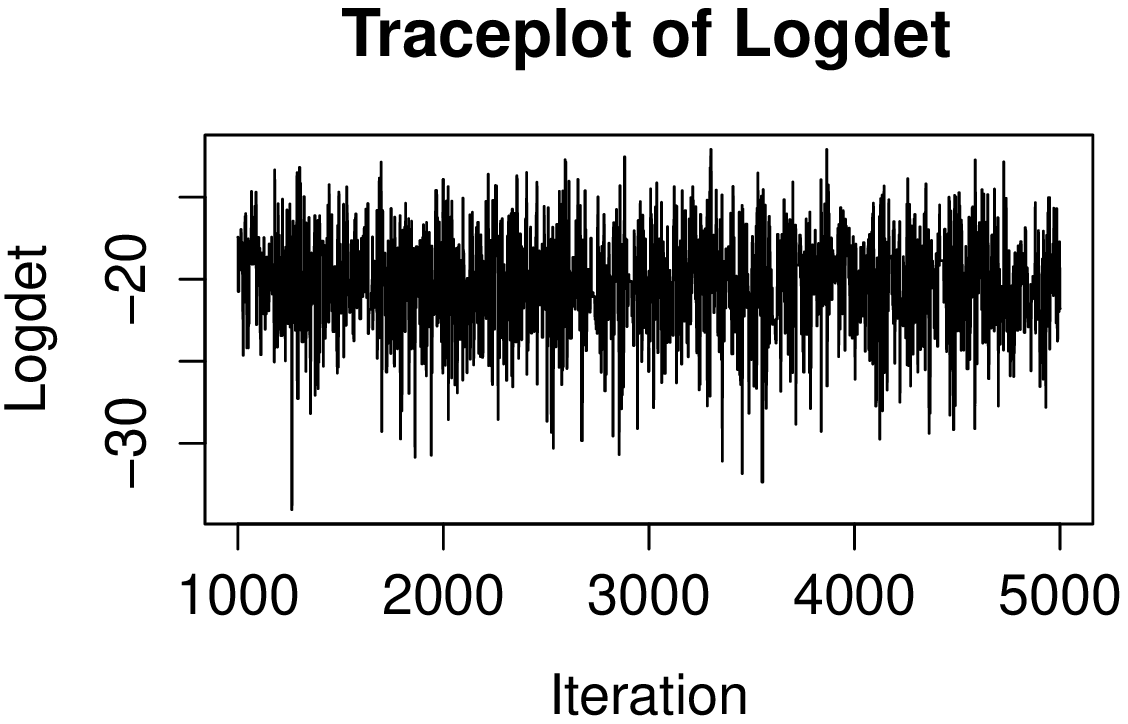}
                 \caption{Traceplot}
                 \label{fig:(b)trace}
         \end{subfigure}%
         ~ 
         \begin{subfigure}[b]{0.3\textwidth}
                 \includegraphics[width=\textwidth]{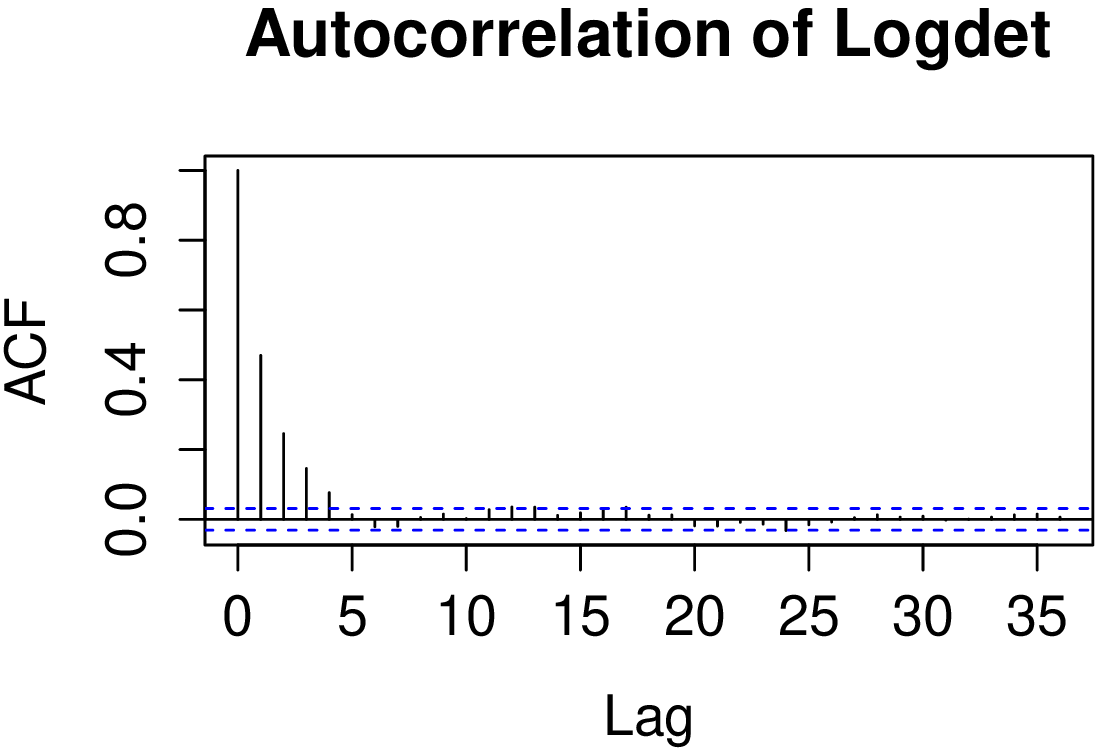}
                 \caption{ACF plot}
                 \label{fig:(b)auto}
         \end{subfigure}

         \caption{Traceplot and Autocorrelation plot of $\log (|K|)$ for Graph in Fig. \ref{fig:1}(b). }\label{fig:TA2}
\end{figure}
\vspace{3mm}

\newpage
\noindent {\bf Graph in Fig.  \ref{fig:1}(c)}
\vspace{3mm}

$D=\left(\begin{array}{cccccccccc}25&9&8&7&6&5&4&3&2&1\\9&25&0&0&0&0&0&0&0&0\\8&0&25&0&0&0&0&0&0&0\\
7&0&0&25&0&0&0&0&0&0\\6&0&0&0&25&0&0&0&0&0\\
5&0&0&0&0&25&0&0&0&0\\4&0&0&0&0&0&25&0&0&0\\3&0&0&0&0&0&0&25&0&0\\2&0&0&0&0&0&0&0&25&0\\1&0&0&0&0&0&0&0&0&25\end{array}\right).$
\vspace{2mm}

\scriptsize

$E(K)=\left(\begin{array}{cccccccccc}0.1229&-0.0013&-0.0026&-0.0039&-0.0052&-0.0065&-0.0078&-0.0091&-0.0104&-0.0117\\-0.0013&0.1229&0&0&0&0&0&0&0&0\\
-0.0026&0&0.1229&0&0&0&0&0&0&0\\
-0.0039&0&0&0.1229&0&0&0&0&0&0\\-0.0052&0&0&0&0.1229&0&0&0&0&0\\
-0.0065&0&0&0&0&0.1229&0&0&0&0\\-0.0078&0&0&0&0&0&0.1229&0&0&0\\-0.0091&0&0&0&0&0&0&0.1229&0&0\\-0.0104&0&0&0&0&0&0&0&0.1229&0\\-0.0117&0&0&0&0&0&0&0&0&0.1229\end{array}\right).$
\normalsize
\vspace{2mm}

\scriptsize
$\hat{K}=\left(\begin{array}{cccccccccc}0.1223&-0.0012&-0.0027&-0.0041&-0.0055&-0.0064&-0.0077&-0.0090&-0.0102&-0.0115\\-0.0012&0.1223&0&0&0&0&0&0&0&0\\
-0.0027&0&0.1223&0&0&0&0&0&0&0\\-0.0041&0&0&0.1223&0&0&0&0&0&0\\-0.0055&0&0&0&0.1223&0&0&0&0&0\\
-0.0064&0&0&0&0&0.1223&0&0&0&0\\-0.0077&0&0&0&0&0&0.1223&0&0&0\\-0.0090&0&0&0&0&0&0&0.1223&0&0\\-0.0102&0&0&0&0&0&0&0&0.1223&0\\-0.0115&0&0&0&0&0&0&0&0&0.1223
\end{array}\right).$

\begin{figure}
         \centering
         \begin{subfigure}[b]{0.3\textwidth}
                 \includegraphics[width=\textwidth]{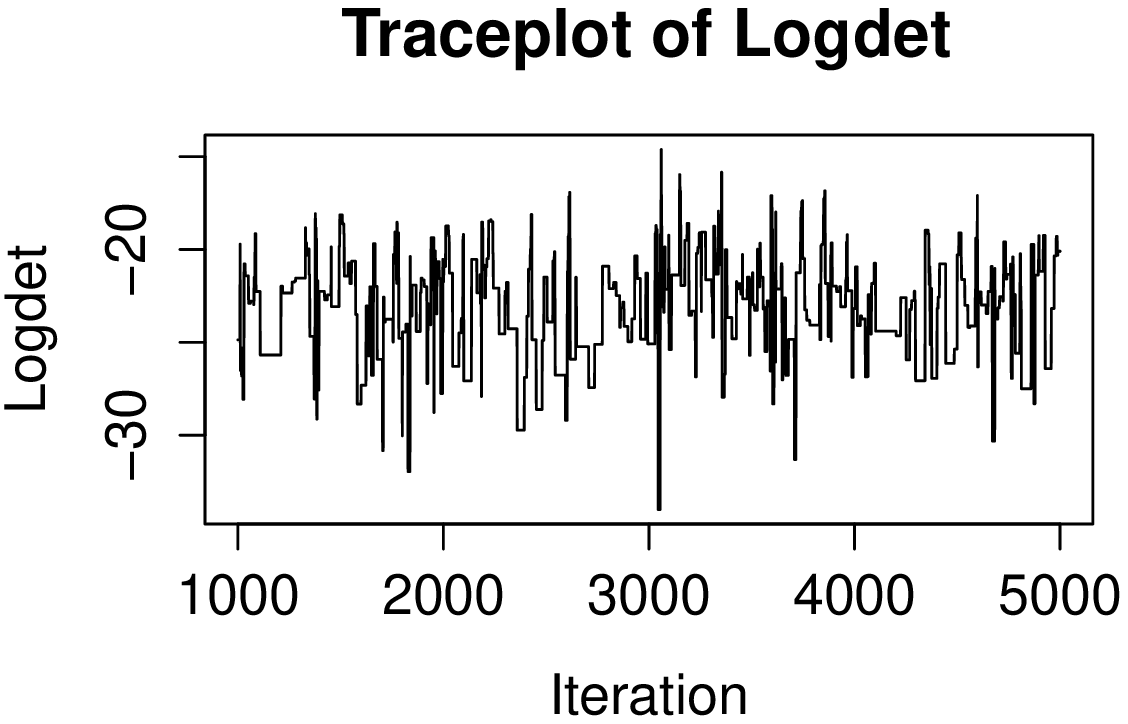}
                 \caption{Traceplot}
                 \label{fig:(c)trace}
         \end{subfigure}%
         ~ 
         \begin{subfigure}[b]{0.3\textwidth}
                 \includegraphics[width=\textwidth]{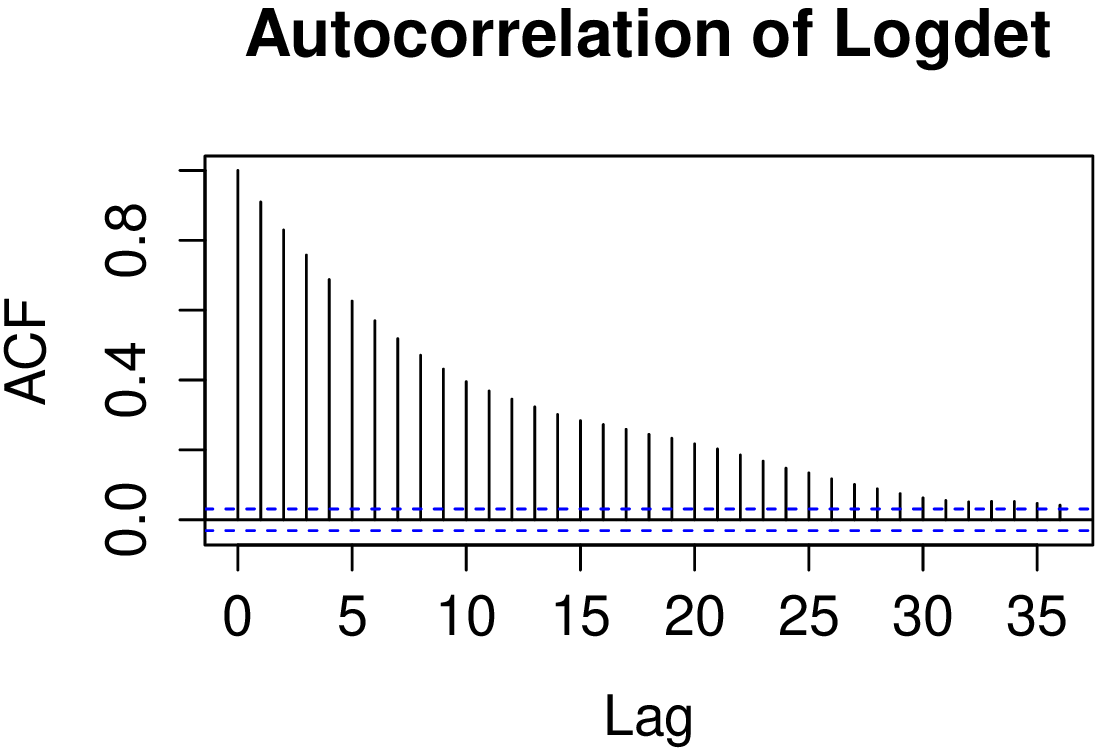}
                 \caption{ACF plot}
                 \label{fig:(c)auto}
         \end{subfigure}

         \caption{Traceplot and Autocorrelation plot of $\log (|K|)$ for Graph in Fig. \ref{fig:1}(c). }\label{fig:TA3}
\end{figure}
\normalsize
\vspace{3mm}

\noindent {\bf Graph  in Fig. \ref{fig:1}(d)}
\vspace{3mm}

$D=\left(\begin{array}{ccc}3&1&2\\1&4&2\\2&2&5\end{array}\right),\;\;\;E(K)=\left(\begin{array}{ccc}1.8108&-0.0073&-0.5517\\-0.0073&1.4472&-0.5517\\-0.5517&-0.5517&1.2413\end{array}\right)$
\newline and
\vspace{2mm}

$\hat{K}=\left(\begin{array}{ccc}1.8097&-0.0075&-0.5514\\-0.0075&1.4485&-0.5514\\-0.5514&-0.5514&1.2442\end{array}\right)$.

\begin{figure}
         \centering
         \begin{subfigure}[b]{0.3\textwidth}
                 \includegraphics[width=\textwidth]{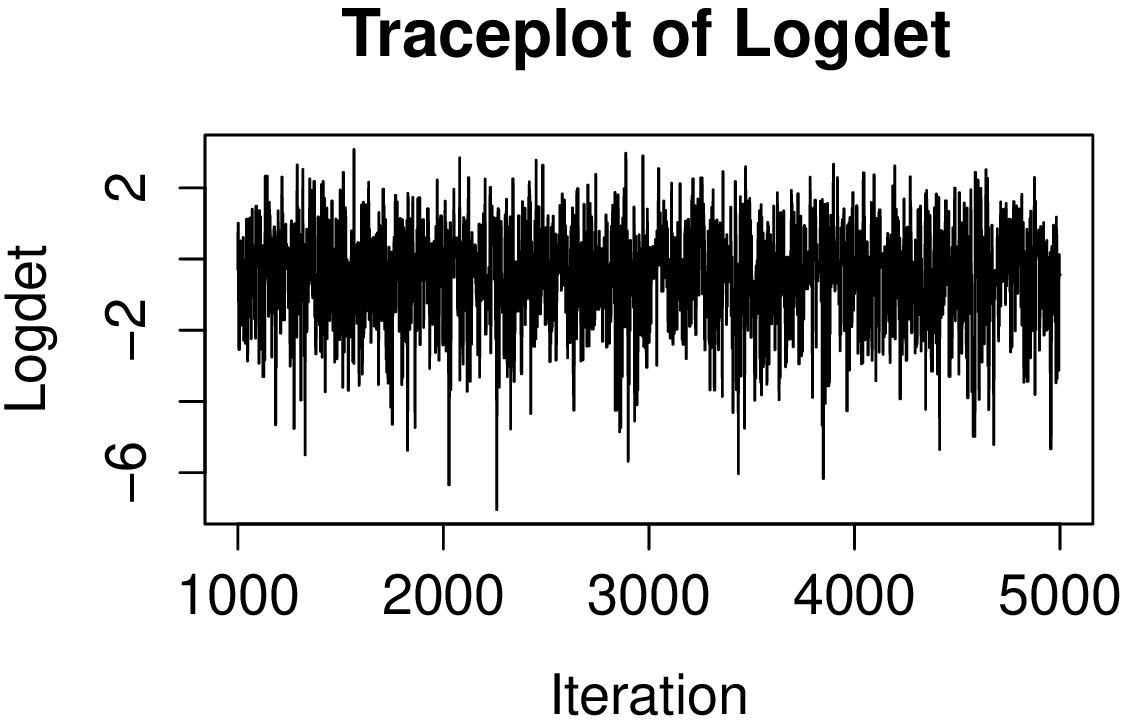}
                 \caption{Traceplot}
                 \label{fig:(d)trace}
         \end{subfigure}%
         ~ 
         \begin{subfigure}[b]{0.3\textwidth}
                 \includegraphics[width=\textwidth]{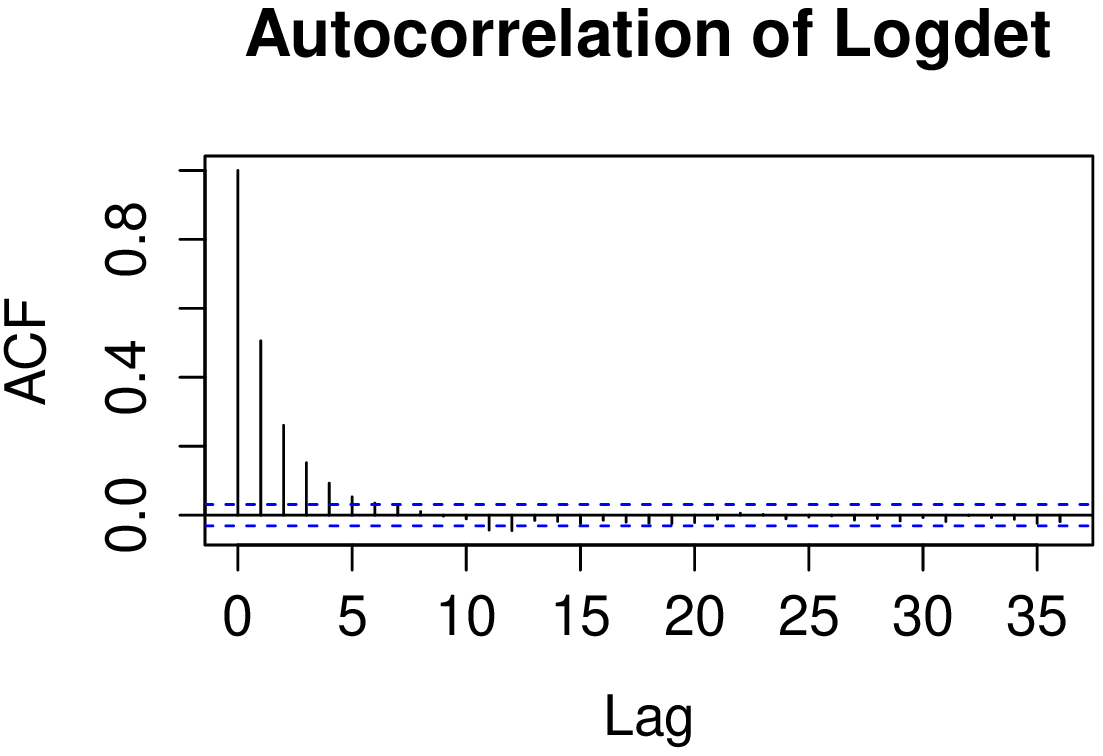}
                 \caption{ACF plot}
                 \label{fig:(d)auto}
         \end{subfigure}

         \caption{Traceplot and Autocorrelation plot of $\log (|K|)$ for Graph in Fig. \ref{fig:1}(d). }\label{fig:TA4}
\end{figure}
\newpage

\noindent {\bf Graph in Fig. \ref{fig:1} (e)}
\vspace{3mm}

$D=\left(\begin{array}{cccc}2&1&3&4\\1&1&3&4\\3&3&200&0\\4&4&0&200\end{array}\right),\;\;\;E(K)=
\left(\begin{array}{cccc}4.4631&-3.5368&-0.0189&-0.0252\\-3.5368&8.4631&-0.0189&-0.0252\\-0.0189&-0.0189&0.0157&0\\-0.0252&-0.0252&0&0.0157\end{array}\right)$
\newline and
\vspace{2mm}

$\hat{K}=\left(\begin{array}{cccc}4.4714&-3.5386&-0.0192&-0.0256\\-3.5386&8.4658&-0.0192&-0.0256\\-0.0192&-0.0192&0.0158&0\\-0.0256&-0.0256&0&0.0158\end{array}\right)$.

\begin{figure}
         \centering
         \begin{subfigure}[b]{0.3\textwidth}
                 \includegraphics[width=\textwidth]{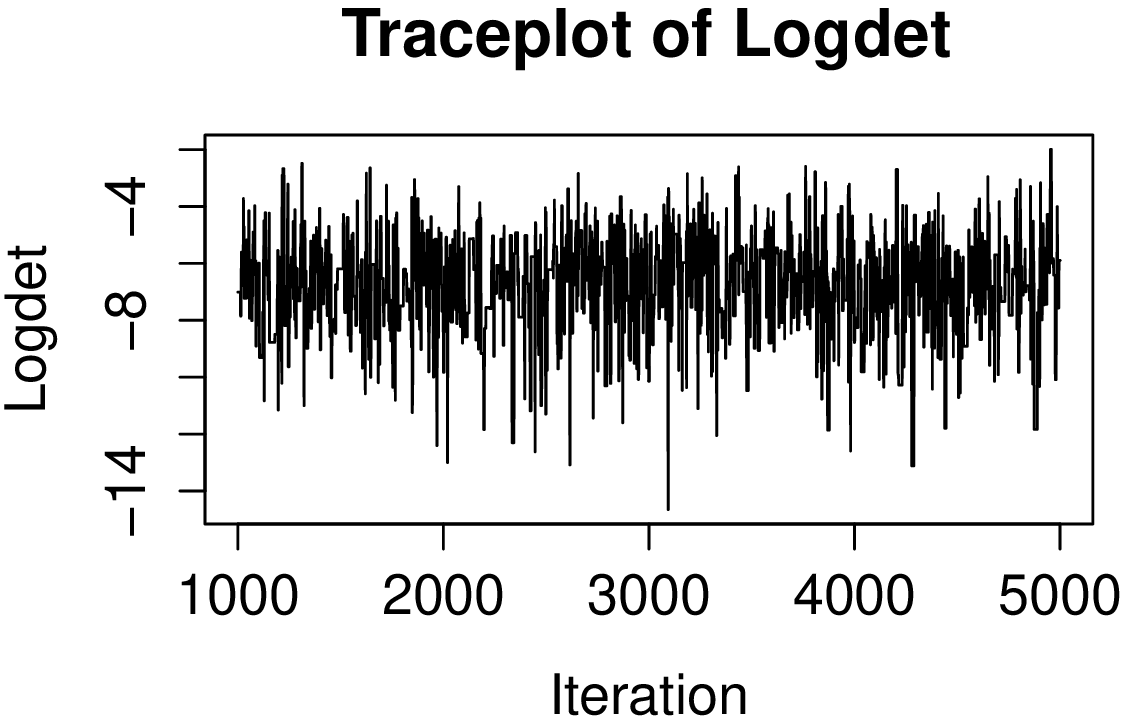}
                 \caption{Traceplot}
                 \label{fig:(e)trace}
         \end{subfigure}%
         ~ 
         \begin{subfigure}[b]{0.3\textwidth}
                 \includegraphics[width=\textwidth]{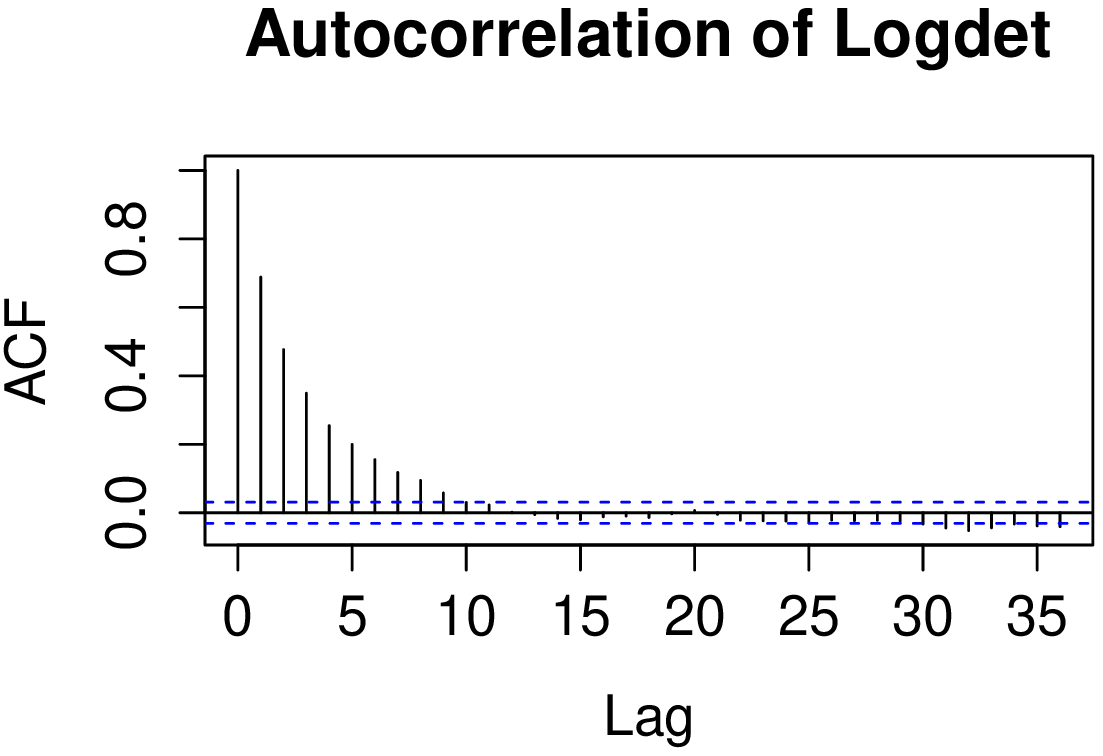}
                 \caption{ACF plot}
                 \label{fig:(e)auto}
         \end{subfigure}

         \caption{Traceplot and Autocorrelation plot of $\log (|K|)$ for Graph in Fig. \ref{fig:1}(e). }\label{fig:TA5}
\end{figure}

\section*{Appendix 3}
\subsection*{Estimates and batch standard errors for entries of $K$ for the models of Section 6}
The estimates  and batch standard errors are given below for the entries of $K$ listed in lexicographic order.

\begin{table}[htbp!]
\centering
\caption{The average estimates for entries of $K$ for Fig.  \ref{fig:7} (b) when $p=20$}
\begin{tabular}{ccccccccc}
         0.1072 &0.0100& 0.0096& 0.0322 &0.0109& 0.0093& 0.0102& 0.0105& 0.0101\\
         0.0103& 0.0099& 0.0106&0.0100 &0.0099 &0.0109& 0.0104 &0.0111& 0.0116\\
         0.0104& 0.0100 &0.0113& 0.0115
   \end{tabular}
\label{table:5}
\end{table}

\begin{table}[h!]
\centering
\caption{The batch standard errors for Fig.  \ref{fig:7} (b) when $p=20$}
\begin{tabular}{ccccccccc}
         0.0004& 0.0005& 0.0005& 0.0001 &0.0005& 0.0005& 0.0005& 0.0005\\
         0.0005& 0.0005& 0.0005 &0.0004& 0.0005& 0.0004 &0.0005& 0.0004\\
         0.0005& 0.0005& 0.0005 &0.0005& 0.0004&0.0005
   \end{tabular}
\label{table:6}
\end{table}

\begin{table}[h!]
\centering
\caption{The average estimates for entries of $K$ for Fig.  \ref{fig:7} (b) when $p=30$}
\begin{tabular}{cccccccccc}
         0.1217& 0.0109& 0.0126& 0.0366& 0.0109& 0.0118& 0.0120& 0.0120& 0.0115& 0.0122\\
         0.0108& 0.0121& 0.0113& 0.0119& 0.0125& 0.0114& 0.0120& 0.0112& 0.0119& 0.0131\\
         0.0115& 0.0125& 0.0116& 0.0132& 0.0110& 0.0119& 0.0119& 0.0107& 0.0129& 0.0119\\
         0.0124& 0.0119\\
   \end{tabular}
\label{table:7}
\end{table}

\begin{table}[h!]
\centering
\caption{The batch standard errors for Fig.  \ref{fig:7} (b) when $p=30$}
\begin{tabular}{cccccccc}
         0.0008& 0.0006& 0.0006 &0.0003& 0.0006& 0.0006& 0.0006 &0.0006\\
         0.0006& 0.0006& 0.0006& 0.0006& 0.0006& 0.0006& 0.0006& 0.0006\\
         0.0006& 0.0006& 0.0006& 0.0006& 0.0006& 0.0006& 0.0006& 0.0006\\
         0.0006& 0.0005& 0.0006& 0.0005& 0.0006& 0.0006& 0.0005& 0.0005\\
   \end{tabular}
\label{table:8}
\end{table}

\begin{table}[h!]
\centering
\caption{The average estimates for entries of $K$ for Fig.  \ref{fig:7} (c) when $p=20$}
\begin{tabular}{cccccccccc}
         0.1102& 0.0106& 0.0104& 0.0347& 0.1135& 0.0329& 0.1104& 0.0335& 0.1113& 0.0332\\
         0.1103& 0.0326& 0.1157& 0.0330& 0.1082& 0.0333& 0.1083& 0.0318& 0.1096& 0.0326\\
         0.1059& 0.0311\\
   \end{tabular}
\label{table:10}
\end{table}

\begin{table}[h!]
\centering
\caption{The batch standard errors for Fig.  \ref{fig:7} (c) when $p=20$}
\begin{tabular}{cccccccc}
         0.0011& 0.0002& 0.0002& 0.0004& 0.0012& 0.0003& 0.0011& 0.0004\\
         0.0012& 0.0004& 0.0012& 0.0003& 0.0013& 0.0003& 0.0012& 0.0003\\
         0.0012& 0.0004& 0.0011& 0.0003& 0.0012& 0.0003\\
   \end{tabular}
\label{table:9}
\end{table}

\begin{table}[h!]
\centering
\caption{The average estimates for entries of $K$ for Fig.  \ref{fig:7} (c) when $p=30$}
\begin{tabular}{cccccccccc}
         0.1295& 0.0117& 0.0111& 0.0384& 0.1253& 0.0386& 0.1266& 0.0376& 0.1248& 0.0357\\
         0.1214& 0.0358& 0.1209& 0.0357& 0.1181& 0.0358& 0.1161& 0.0349& 0.1126& 0.0345\\
         0.1123& 0.0339& 0.1126& 0.0338& 0.1136& 0.0330& 0.1143& 0.0323& 0.1083& 0.0324\\
         0.1077& 0.0318\\
   \end{tabular}
\label{table:11}
\end{table}

\begin{table}[h!]
\centering
\caption{The batch standard errors for Fig. \ref{fig:7} (c) when $p=30$}
\begin{tabular}{cccccccccc}
         0.0013& 0.0002& 0.0002& 0.0003& 0.0012& 0.0004& 0.0011& 0.0004\\
         0.0011& 0.0003& 0.0011& 0.0003& 0.0013& 0.0004& 0.0010& 0.0004\\
         0.0011& 0.0003& 0.0010& 0.0003& 0.0010& 0.0003& 0.0011& 0.0003\\
         0.0010& 0.0003& 0.0012& 0.0003& 0.0010& 0.0003& 0.0010& 0.0003\\
   \end{tabular}
\label{table:11}
\end{table}

\end{document}